\documentclass[a4paper,12pt]{article}
\usepackage{geometry}

\usepackage{csquotes}
\usepackage[section]{placeins}%keep pics in correct section
\usepackage[english]{babel}
\usepackage[utf8]{inputenc}
\usepackage[T1]{fontenc}
\usepackage{multirow}
\usepackage{listings}
\usepackage{verbatim}

\usepackage{enumerate}
\usepackage{enumitem}
\usepackage{standalone}
\usepackage[url=false,isbn=false,giveninits=true]{biblatex}
\bibliography{sources}
\usepackage[utf8]{inputenc}
\usepackage[english]{babel}
\usepackage[T1]{fontenc}
\usepackage{amsmath, amsthm, amssymb}
\usepackage{mathtools}
\usepackage{amsfonts}
\usepackage{amssymb}
\usepackage{graphicx}
\usepackage{braket}
\usepackage{relsize}
\usepackage{comment}
\usepackage{pifont}
\usepackage[makeroom]{cancel}
\usepackage[absolute,overlay]{textpos}
\usepackage{scalerel}
\usepackage{xcolor}
\usepackage{changepage}
\usepackage{needspace}

\usepackage[font=small,labelfont=bf,tableposition=top]{caption}
\usepackage{url,amssymb,floatrow,graphicx,amsmath,upgreek,textcomp,tabularx,subcaption,caption,gensymb,float}  %fullpage
\usepackage{wrapfig,svg}
\usepackage{array, makecell, booktabs}
\usepackage{float}
\usepackage[absolute,overlay]{textpos}
\usepackage[nottoc,numbib]{tocbibind}
\usepackage{url}
\usepackage{hyperref}
\hypersetup{
    colorlinks=true,
    linkcolor=red,
    citecolor=blue,
    urlcolor=blue,
}
\usepackage{subfiles} % Best loaded last in the preamble

\newtheorem{proposition}{Proposition}
\newtheorem{corollary}{Corollary}
\newtheorem{theorem}{Theorem}

\newtheorem{conjecture}{Conjecture}
\theoremstyle{definition}
\newtheorem{definition}{Definition}
\newtheorem{example}{Example}

\newtheorem{remark}{Remark}

\DeclarePairedDelimiter\inner{\langle}{\rangle}
\DeclarePairedDelimiter\abs{\lvert}{\rvert}
\DeclarePairedDelimiterX{\norm}[1]{\lVert}{\rVert}{#1}
\DeclarePairedDelimiterX{\parr}[1]{(}{)}{#1}

\DeclareMathOperator{\Tr}{Tr}

\DeclareMathOperator{\Image}{Im}
\DeclareMathOperator{\diag}{diag}
\DeclarePairedDelimiter\br{\langle}{\rvert}
\DeclarePairedDelimiter\ke{\lvert}{\rangle}
\DeclareMathOperator{\povm}{POVM}

\DeclareMathOperator{\sa}{sa}

\begin{document}
\newgeometry{right=0cm,left=0cm,top=2.0cm,bottom=1cm}

%=============================================================== TITLE PAGE
\clearpage  
\thispagestyle{empty}

\begin{center}
   \LARGE{
   \vspace*{10mm}
   DIPLOMA THESIS\\\vspace{30mm}
   On a conjecture regarding quantum hypothesis\\ testing\\\vspace{30mm}
   Zsombor Szilágyi
   }
\\\vspace{60mm}
%An example of quantum hypothesis testing
\end{center}

\begin{table}[H]
 \large{
\begin{tabular}{lll}
\hspace{40mm} Supervisor: \hspace{10mm}  & Mihály Weiner  &  \\
            & associate professor  &  \\
            & BME Institute of Mathematics & \\ 
            & Department of Analysis & \\\vspace*{10mm}
\end{tabular}
}
\end{table}

\begin{center}
     \large{BME \\ 2020}
\end{center}

\restoregeometry

\newpage
\setcounter{tocdepth}{2}
{
  \hypersetup{linkcolor=black}
  \tableofcontents
}
\pagebreak

\setlength{\parindent}{0em}
\setlength{\parskip}{\medskipamount}  %bekezdesek miatt

%\subfile{sections/1_intro}
%\subfile{sections/2_single-shot}
%\subfile{sections/3_asympt}
%\subfile{sections/4.0}
%\subfile{sections/4_conj}

%=============================================================== PREFACE
\section*{Preface}
\addcontentsline{toc}{section}{\protect\numberline{}Preface}

The first time I got introduced to a more mathematical point of view of quantum physics was when I studied \textit{Functional Analysis} from Mihály Weiner at BME (in the spring of 2018), who my supervisor is now. I already had a deep interest in quantum physics, but with the beautiful mathematics involved it become even more exciting to me. In the next semester I attended lectures in \textit{Matrix analysis} and \textit{Quantum probability theory} - this latter was a course Mihály Weiner gave for American students at Budapest Semesters in Mathematics (BSM). In the same semester I also attended a course called \textit{Introduction to quantum information theory} presented by Milán Mosonyi, which introduced me to the finite dimensional quantum statistical physics in a very detailed and mathematically rigorous way. The next semester I attended the second part of this course 
(\textit{Selected topics in modern quantum information theory}), where I first met the topic of \textit{quantum hypothesis testing} (or \textit{quantum state discrimination}) and the multitude of unsolved questions about it. I also attended a course on quantum computers offered by our Department of Theoretical Physics  (\textit{Quantum Information Processing}) where we could use remotely accessible quantum computers and could experience how qubits work in practice.\\

What became my area of interest is the surprising and extremely fruitful interaction of quantum physics and the part of mathematics usually referred to functional analysis (and/or matrix-analysis). I am therefore very thankful to all of my teachers I met during my university studies; both to those, who helped me to acquire a good physics background and also those who provided me the necessary mathematics to my chosen subject. I am particularly grateful to Mihály Weiner and Milán Mosonyi for the great motivating lectures, Péter Bálint and Péter Lévay who first introduced me to the beauty of matemathics, László Szunyogh and Balázs Dóra who taught me quantum mechanics, Gergely Zaránd and János Török for the great classical mechanics and statistical physics lectures, András Pályi and László Orosz for the instructive sessions. I would also like to use the occasion to thank my parents and my cousins who supported and encouraged me during my university studies. Finally, I would like to gratefully acknowledge the financial support provided by the MTA-BME Quantum Information Theory Research Group, and by the NRDI grants K124152 and KH129601. \\

During the quantum probability theory course I frequented at BSM I got to know a math student (from Cornell University) who was also interested in quantum information theory. We started to work on a problem introduced by Mihály Weiner to us. This was a very interesting question regarding \textit{mutually unbiased bases} (MUBs), which are matrix-algebraic structures naturally arising in various quantum informational protocols (e.g., in some cryptographic one). The research went well, and together with Mihály Weiner we managed to answer the question which eventually resulted in the paper \textit{Rigidity and a common framework for mutually unbiased bases and k-nets} \cite{mutually} recently accepted in the prestigious Journal of Combinatorial Designs. \\

The above research was also (one of) the subjects indicated in my original thesis topic. However, after
resolving the issue and still having several months till the deadline of the thesis submission, I decided that rather than just reporting on a completed work, I should 
look for new questions and - if I am successful - make my thesis on my latest investigations. Still working on the general area of quantum information theory, on a seminar held by Milán Mosonyi I became interested on a problem regarding quantum hypothesis testing. The aim was to prove or disprove a conjecture raised several years ago \cite{Milan} about the best achievable asymptotic error rate in \textit{composite hypothesis testing}. Following the suggestion of my supervisor Mihály Weiner I first considered several particular cases proposed by him where I could prove (or find a counterexample to) the conjecture. While at this point I cannot report on a complete success, I managed to resolve the issue
in one of the interesting cases proposed. I therefore decided to slightly change my originally declared research topic (and consequently also the title) and center my thesis work around the topic explained above. So in this thesis I will give an introduction to quantum hypothesis testing, explain the conjecture in question, the importance of the particular case I considered, and my proof regarding this case.
\\

\begin{comment}
In Section \ref{sec:Int} we give an outline of hypothesis testing and the aim of this thesis. In Subsection \ref{sec:Not} we only fix the basic concepts and clarify the notation, it assumes the basic knowledge of the subject (for a complete introduction see \cite{QInotes}).

In Section \ref{sec:Sin} we give a brief introduction to the single-shot state discrimination. The goal is here to give a background to understand the conjecture in the next section. (For a comprehensive study of the subject see \cite{QInotes})

In Section \ref{sec:Asy} we give a brief introduction to the asymptotic state discrimination and state the conjecture which is the main point of this thesis (for the original article  about the conjecture see \cite{Milan}). After that we show the known special cases. In Subsection \ref{sec:Ver}  we present the new case and explain why it is exciting, and give a proof which confirms the conjecture in this special case (and this is the essence of my work).
\end{comment}

\newpage
%=============================================================== INTRO
\section{Introduction}\label{sec:Int}
\subsection{What is hypothesis testing, and why is it important?}

Roughly speaking, \textit{hypothesis testing} (or \textit{state discrimination}) is the problem of distinguishing states (of a physical system) from each other. This problem appears at the end of any communication or computation (or measurement) process, where the resulting set of possible outcomes is described by a set of states of the system and one has to identify the true outcome. So the relevance of this task is obvious. The physical system could be either classical or quantum, but in this thesis we are focusing to the quantum case.

In more detail, the easiest way is to represent the problem with the following  communication scenario. Consider that a sender (Ailce) wants to communicate
one of $r$ possible messages to a receiver (Bob). To achieve this goal, they use a $d$-level quantum physical system (i.e., one whose Hilbert space is $d$-dimensional). Alice can encode her message $i\in[r]$ into a state $\varrho_i$ of the system and send it to Bob, who wants to decide which message he got. Bob, knowing that the only possible states are the previously agreed operators $\varrho_i$, tries to perform a measurement which gives the correct answer with the best probability. One may assume that the message of Alice is the value of a classical $r$-bit, which is in fact a classical random variable taking values in the set $[r]=\{1,\ldots ,r\}$ with corresponding probabilities $p_1,\ldots ,p_r$. The situation can be visualised as follows. \\
\begin{table}[H]
\begin{tabular}{cccc}
 &\textit{Alice} &  Q channel = id    & \textit{Bob}    \\
 & CL r-bit   & $\varrho_i \ \longrightarrow \ \varrho_i$ & CL r-bit      \\[10pt]  \hline\\[-5pt] 
$p_1 \quad$                   & 1: $\varrho_1$                     &                                                           & 1: $M_1$     \\[10pt]
$p_2 \quad $               & 2: $\varrho_2$                     &                                                           & 2: $M_2$     \\[10pt]
               &   $\vdots$                     &                                                           &   $\vdots$     \\[10pt]
$p_r \quad $               & $r$: $\varrho_r$                     &                                                           & $r$: $M_r$     \\[10pt]
\end{tabular}
\end{table}
Here the operators $\varrho_i$ are the encoding states and $M_i$ are the positive operators forming a POVM which describe (from point of view of outcome probabilities) the particular measurement procedure chosen by Bob.

Ideally, they should use states which are distinguishable with certainty (this is called \textit{perfect state discrimination}). However, the state set by Alice may be corrupted ``on the way'', e.g., by some noise in the channel of communication. Here, instead of considering noisy channels, we shall formally regard our channel a perfect one (i.e., it is the identity map), but restrict Alice from using arbitrary states for encoding. To put it another way, we shall simply ``prescribe'' to Alice which states she must use. So it makes sense to consider not only perfectly distinguishable states, but ones which cause mistake in decoding with a non-zero probability. To minimize this error is the subject of \textit{minimum error state discrimination}.

A natural way to reduce error is to send the message several times. It is known that the error vanishes exponentially fast in the number of repetitions. The aim of \textit{asymptotic state discrimination} is to determine this exponent. In the case of \textit{binary state discrimination}, i.e., when $r=2$ (so there are only two possible messages), the exponent is the well-known \emph{Chernoff divergence} of the states $\varrho$ and $\sigma$ (denoted by $C(\varrho,\sigma)$).

%=============================================================== AIM OF THIS THESIS
\subsection{Composite hypothesis testing and the aim of this thesis}

In realistic scenarios it is more natural to assume that the hypotheses are represented by sets of states. This is because Bob might only need to know whether the message falls into certain subsets. Say for example, that the message is one of the values $i$ of the set $\{1,2,3,4,5\}$, but Bob actually only needs to know if $i<4$ or not. Of course, if Alice also knew what Bob needs, then she could just use a smaller alphabet, but we assume that she is unaware of the needs of Bob (or equivalently, that the need of Bob only becomes clear after the message was sent).  This is called \textit{composite hypotheses testing}.

In the \textit{single shot case} (i.e., when the message is not repeated) the above composite hypothesis testing can be reduced to a non-composite one. Indeed, instead of saying that depending on the value of $i\in \{1,2,3,4,5\}$
Alice puts the system in state $\rho_i$ and Bob only needs to decide whether $i<4$ or not, we could just say that Alice either sends a certain convex combination $\tilde{\rho}$ of  $\rho_1,\rho_2,\rho_3$ \textit{or} another certain combination $\tilde{\sigma}$ of $\rho_4,\rho_5$. However, in the repeated case, 
the communicational situation we considered -- namely, that Alice does not know what Bob exactly needs -- leads to the mathematical problem of distinguishing a certain convex combination of $\rho_1^{\otimes n},\rho_2^{\otimes n},\rho_3^{\otimes n}$
from another certain combination of $\rho_4^{\otimes n},\rho_5^{\otimes n}$
rather than distinguishing $\tilde{\rho}^{\otimes n}$ 
and $\tilde{\sigma}^{\otimes n}$ from each other. 
Thus, the computation of the asymptotic error exponent in the composite case cannot be simply reduced to a computation regarding the non-composite one.

It is, however, easily seen that also in this case the error vanishes exponentially fast with the number of repetitions, but the exponent is unknown even in the binary case. In fact, already in the simplest possible setup, namely, when the first set consists only one state $\{\varrho\}$ and the second consists only two states $\{\sigma_1,\sigma_2\}$, the exponent is unknown.
There is a conjecture \cite{Milan} that the exponent in this case is the minimum of the Chernoff divergences (i.e., $\min_i C(\varrho,\sigma_i)$) of the pairs (``worst case exponent''), that is the composite error decreases as the slower decreasing of the two pairs.

In some special cases, analytical proofs show that the conjecture is true there. Some numerical computations done in low dimensional cases also suggest that the conjecture is true. However, all of the known special cases are in some sense ``nice'', e.g., having certain symmetries. This may make the statement of the conjecture true for them, for example, because of some underlying group structure. Numerical computations can also be misleading: the deviation from the conjectured value could be very small, and to get high precision asymptotic rates, one needs to consider high tensorial powers, which means, that even if we start in low dimensions, one needs computations with extremely large matrices. Therefore, it could happen that in a less ``nice'' case the conjecture fails. This was precisely what happened in the well-known question of ``\textit{Superadditivity of communication capacity using entangled inputs}'' \cite{Hastings2009}, where a certain  natural conjecture was made, and for a long time many great researchers were attempting to prove it in vain. Numerical searches did not find counterexamples, and in all special cases that could be easily considered, the conjecture was true. However, in the end, the conjecture was shown to be false by a random counterexample. So the reason people could not find a counterexample was that the cases that are easy to deal with are just ``too nice''.   

Because of the above, one is motivated to find cases which on the one hand are as ``asymmetrical'' as possible (and, of course, do not follow from any of the known cases), yet still analytically computable. Following a suggestion of my supervisor, I will consider such a new case. As it turns out in my diploma thesis, the conjecture is also true in this new special case, which increases the likelihood that the conjecture is true in general.

In my thesis, during the introduction of the above topics (until the last subsection, which is my work) I omit the proofs of the known theorems and propositions, because it is not crucial for getting the essence of my thesis. However, if the curious reader interested in the proofs, they can find them in \cite{QInotes, Aud1}.

%_______________________________________________________________________________ Notations
\subsection{Notations and basic concepts}\label{sec:Not}

The aim of this section is only to fix the basic concepts and clarify the notation. It assumes the basic knowledge of the subject, and it is not intended to provide an introduction to it.
(For a complete study of the subject, see \cite{QInotes}.)

%________________________________________________________________________________ Q
\subsubsection{Quantum states and measurements}
We will only deal with finite dimensional quantum systems, where by the \textit{dimension} of the system we mean the dimension of the Hilbert space. Let $\mathcal{B}(\mathcal{H})$ denote the linear operators on a finite-dimensional Hilbert space $\mathcal{H}$. The set of self-adjoint operators $\mathcal{B}(\mathcal{H})_{\textrm{sa}}$ forms a real subspace in $\mathcal{B}(\mathcal{H})$. The set of positive semidefinite (PSD) operators $\mathcal{B}(\mathcal{H})_+$ forms a convex cone in $\mathcal{B}(\mathcal{H})_{\textrm{sa}}$. We use the notation $\varrho \geq0$ to express that $\varrho$ is a PSD operator. We define the \textit{quantum state space} as
\[
\mathcal{S}(\mathcal{H}):=\{\varrho \in \mathcal{B}(\mathcal{H})\ :\ \varrho \geq0, \ \Tr\varrho=1 \},
\]
that is the intersection of the hyperplane (of the trace one operators) and the convex cone $\mathcal{B}(\mathcal{H})_+$, so it is a convex set. The elements $\varrho\in\mathcal{S}(\mathcal{H})$  are called \textit{quantum state}s or \textit{density operator}s.
We use Dirac's great bra–ket notation $\ket{\psi}\bra{\psi}$ (where $\psi\in\mathcal{H}, \norm{\psi}=1$) to denote the rank one orthogonal projection to the one dimensional subspace $\mathbb{C}\psi$. These are the \textit{pure state}s and they are precisely the extreme points of the state space $\mathcal{S}(\mathcal{H})$. A \textit{nontrivial} (i.e., $\dim \mathcal{H} \geq 2$) quantum system have continuum many pure states. Since a density operator $\varrho$ is PSD, it is in particular also self-adjoint, and hence, normal, and therefore it can be expanded as 
\[
\varrho=\sum_{i=1}^d \lambda_i \ket{e_i}\bra{e_i},
\]
where $\lambda_1,...,\lambda_d$ are the eigenvalues of $\varrho$ and the eigenvectors $e_1,...,e_d$ form an ONB in $\mathcal{H}$ (this is what we call \textit{eigendecomposition} or \textit{spectral decomposition}). From the two condition $\Tr\varrho=1$ and $\varrho \geq0$, it follows that the eigenvalues form a probability distribution on $d=\dim \mathcal{H}$ points.

We consider $\mathcal{B}(\mathcal{H})$ as an inner product space with the \textit{Hilbert-Schmidt inner product} defined as follows,
\[
\braket{A,B}_{\textrm{HS}}:=\Tr A^* B.
\]
We say that $A,B\in\mathcal{B}(\mathcal{H})$ are \textit{orthogonal}, and denote by $A\perp B$, if \[
\braket{A,B}_{\textrm{HS}}=0.
\]
The following Proposition shows several equivalent characterizations of the orthogonality of operators.
\begin{proposition}\label{prop:ort}
Let $A,B\in\mathcal{B}(\mathcal{H})_+$, then the following are equivalent:
\begin{enumerate}[label=(\roman*)]
\item $A\perp B$, i.e., $\braket{A,B}_{\textrm{HS}}=0$,
\item $\Tr A B =0$,
\item $ A B =0$,
\item $\Image A\perp \Image B$,
\item $A^0 B^0=0$,
\end{enumerate}
where $A^0$ denotes the orthogonal projection to $\Image A$.
\end{proposition}
We say that an operator $A\in\mathcal{B}(\mathcal{H})$ is an \textit{orthogonal projection}, if $A=A^2=A^*$. 
\begin{remark}
In the whole thesis, by \textit{orthogonal projection} we  mean this latter definition (even in the phrase ``$A,B$ are two orthogonal projections''). If we want to express that they are orthogonal to each other we use $A\perp B$ or say ``...operators which are orthogonal''.
\end{remark}

On a \textit{quantum measurement} we mean a \textit{positive operator valued measure} (POVM) on a Hilbert space $\mathcal{H}$ which is a set of positive operators $M_{x_1},...,M_{x_r}$, sum up to the identity,
\[
\sum_{x_i\in\mathcal{X}} M_{x_i}=I,
\]
where the indices denote the outcomes of the set $\mathcal{X}$. The operators $ M_{x}$ are called \textit{measurements operators} of the given POVM.
A POVM is called a \textit{projection valued measure} (PVM) if all $M_x$ are orthogonal projections (i.e., $M_x=M_x^2=M_x^*$). The POVM forms a convex set (with natural pointwise operations).

The \textit{generalized Born rule} gives the probability of the outcome $x$ in the state $\varrho$ for the measurement $M$,
\[
P_{\varrho,M}(x):=\Tr \varrho M_x,
\]
for any state and any measurement.

In the case of basic quantum mechanics, when we restrict ourselves to pure states (instead of density operators) and PVM (instead of POVM), that is, we have the state $\varrho=\ket{\psi}\bra{\psi}$ and the PVM $\{\ket{e_1}\bra{e_1},...,\ket{e_r}\bra{e_r}\}$, the Born rule returns to the well-known formula
\[
P_{\varrho,M}(i)=\Tr \Big(\ket{\psi}\bra{\psi}\ket{e_i}\bra{e_i}\Big)=\abs{\braket{e_i,\psi}}^2,
\]
where we used the cyclic property of the Tr.

%________________________________________________________________________________ C
\subsubsection{Classical systems in quantum formalism}
Let $\mathcal{H}$ be a finite-dimensional Hilbert space and fix an orthonormal basis $\beta: e_1,...,e_r$ in $\mathcal{H}$. Let $\mathcal{B}(\mathcal{H})^{\beta}$ denote the linear operators on $\mathcal{H}$ which have diagonal matrix in the $\beta$ basis. That is, for any $A\in \mathcal{B}(\mathcal{H})^{\beta}$,
\[
A=\sum_{\omega\,\in\,[r]} A(\omega) \ket{e_\omega}\bra{e_\omega},
\]
where $A(\omega)$ denotes the ith diagonal element (so the ith eigenvalue). We denote by $\mathcal{B}(\mathcal{H})_{\textrm{sa}}^{\beta}$ and $\mathcal{B}(\mathcal{H})_+^{\beta}$ the set of self-adjoint operators and the set of positive semidefinite operators in $\mathcal{B}(\mathcal{H})^{\beta}$, respectively. We define the \textit{classical state space} as
\[
\mathcal{S}(\mathcal{H})^{\beta}:=\{\varrho \in \mathcal{B}(\mathcal{H})^{\beta}\ :\ \varrho \geq0, \ \Tr\varrho=1 \}.
\]
The elements $\varrho\in\mathcal{S}(\mathcal{H})^{\beta}$  are called \textit{classical state}s. We define the \textit{pure state}s in the same way as in the quantum case. The big difference is that a classical system can only have finite ($r$) many pure states. The classical state space forms an (r-1)-simplex.

By a \textit{classical measurement} we mean a POVM where all the measurement operators $M_{x_1},...,M_{x_r}$ are diagonal in the $\beta$ basis.

\begin{remark}
It is easy to show that, if we have a classical system than it is enough to consider only the classical measurements. We can not get any new probability distribution on the outcomes, if we consider the non-diagonal ones, so quantum measurements do not offer any benefit in the classical case.
\end{remark}

%__________________________________________________________________________ asy not
\newpage
\subsubsection{Asymptotic notation}

Let $f$ and $g$ be real valued functions defined on some unbounded subset of the positive real numbers, and $g(n)$ be strictly positive for all large enough values of $n$. We define the following notations in the table (where the limit definitions assume $g(n)\neq 0$ for sufficiently large $n$).\\

\setlength{\tabcolsep}{10pt}
\begin{table}[H]
\begin{adjustwidth}{-2cm}{}
\begin{tabular}{l|l|l|l|l}
 & Notation & Name & Formal definition & Limit definition  \\[5pt]\hline
 \rule{0pt}{5ex}    $<$ & $f=o(g)$ & Small O &  $\forall k>0\ \ \exists n_0\ \ \forall n>n_0:\  $ & $\displaystyle \lim_{n \to \infty} \frac{\abs{f(n)}}{g(n)}=0$ \\
   &   &   & $\abs{f(n)}<k g(n)$ \\[15pt]
$\leq$ & $f=O(g)$ & Big O & $\exists k>0\ \ \exists n_0\ \ \forall n>n_0: $ & $\displaystyle \limsup_{n \to \infty}\frac{\abs{f(n)}}{g(n)}<\infty$ \\ 
   &   &   & $\abs{f(n)}\leq kg(n)$ \\[15pt]
$\approx$ & $f=\Theta(g)$  & Big Theta & $\exists k_1>0\ \ \exists k_2>0\ \ \exists n_0\ \ \forall n>n_0:$ & $f(n)=O(g(n))\ $ and    \\ 
  &   &   & $k_1 g(n)\leq f(n) \leq k_2 g(n)$ & $f(n)=\Omega(g(n))$ \\[15pt]
$\sim$ & $f\sim g$ & On the order of & $\forall \epsilon >0\ \ \exists n_0\ \ \forall n>n_0: $ & $\displaystyle \lim_{n \to \infty}\frac{f(n)}{g(n)}=1 $ \\ 
   &   &   & $\abs{\frac{f(n)}{g(n)}-1}<\epsilon$ \\[15pt]
$\geq$ & $f=\Omega(g)$ & Big Omega & $\exists k>0\ \ \exists n_0\ \ \forall n>n_0: $ & $\displaystyle \liminf_{n \to \infty}\frac{f(n)}{g(n)}>0 $ \\
   &   &   & $f(n)\geq k g(n)$ \\[15pt]
$>$ & $f=\omega(g)$ & Small Omega & $\forall k>0\ \ \exists n_0\ \ \forall n>n_0: $ &  $\displaystyle \lim_{n \to \infty}\Big|\frac{f(n)}{g(n)}\Big|=\infty $   \\
   &   &   & $\abs{f(n)}> k \abs{g(n)}$ \\[5pt]
\end{tabular}
\end{adjustwidth}
\end{table}
In the first column there are some alternative notations, for example, instead of $f=o(g)$ we can use $f<g$, if it is clear from the context. In this work we will use frequently the notation $o$ and the notations in the first column $<,\leq,\approx,\sim,\geq,>$.
\begin{remark}\label{re:ordo}
\par
\needspace{4\baselineskip}
So far this is consistent with the ordinary notations. The only slight change what we will make is the following. In the proofs, with the above notation, we only want to express the order of the functions regardless to their signs. So, for example, when we write $f=o(g)$, it actually means that $\abs{f}=o(\abs{g})$.
\end{remark}

\newpage

%=============================================================== SINGLE-SHOT
\section{Single-shot state discrimination}\label{sec:Sin}
Consider the following communication scenario (what we have already introduced in the previous section). A sender (Alice) has $r$ given states $\varrho_1,...,\varrho_r$ and a probability distribution $p_1,...,p_r$, which are known for Bob (both the states and the probability distribution). Alice randomly chooses a state (with respect to the probability distribution) and send it to Bob through an ideal quantum channel (i.e., Bob gets exactly the same state what Alice sent). The task for Bob is to find out which message was sent. For this, Bob performs a measurement on the system, described by a POVM. The point is that Bob has the prior knowledge that the message can only be one of the $r$ given states $\varrho_1,...,\varrho_r$ with the probability $p_1,...,p_r$, and the only freedom is in choosing his measurement. This problem is called the \emph{single-shot state discrimination} (where the single-shot refers to that the state has been sent only once, unlike in the asymptotic state discrimination, where the state has been sent n times). The situation can be visualised as follows.

\begin{table}[H]
\begin{tabular}{clll}
\multicolumn{1}{l}{} & \multicolumn{1}{l}{\textit{Alice}} & \multicolumn{1}{c}{$\varrho_i \ \longrightarrow \ \varrho_i$} & \textit{Bob}    \\[10pt]
$p_1 \quad$                   & 1: $\varrho_1$                     &                                                           & 1: $M_1$     \\[10pt]
$p_2 \quad $               & 2: $\varrho_2$                     &                                                           & 2: $M_2$     \\[10pt]
               &   $\vdots$                     &                                                           &   $\vdots$     \\[10pt]
$p_r \quad $               & $r$: $\varrho_r$                     &                                                           & $r$: $M_r$     \\[10pt]
\end{tabular}
\end{table}

If Bob obtains the outcome $j$, he concludes that Alice has sent the message $j$.
By the Born rule, the probability of that he concludes the message being $j$ when Alice sent $i$ is
\[
P(j|i)=\Tr M_j \varrho_i.
\]
Thus, the probability of successfully identifying the message $i$ is
\[
P(i|i)=\Tr M_i \varrho_i,
\]
the \textit{i-th success probability}, while the probability of an incorrect decoding is
\[
1-P(i|i)=\Tr(I-M_i) \varrho_i,
\]
the \textit{i-th error probability}.
\begin{remark}
A $d$-dimensional classical system has $d$ different pure states, while any nontrivial (i.e., $d\geq2$ dimensional) quantum system has continuum many pure states. In particular, a classical bit has 2 different pure states 0 and 1, while a qubit has continuum many (every surface-point of the Bloch sphere). Thus specifying the physical state of a classical bit requires 1 bit of information, on the contrary, specifying the physical state of a qubit requires an infinite sequence of 0-1. This might suggest that it may be possible to store an infinite amount of information in a single qubit. Of course, it is not possible in practice, since preparing a qubit exactly in the state and reading it out would require infinite precision. In the following we will see that it is not even possible theoretically to store more information in a finite-dimensional quantum system than in a classical one (of the same dimension).
\end{remark}

%_______________________________________________________________ PERFECT
\subsection{Perfect state discrimination}
In this subsection we examine when it is possible for Bob to perform such a measurement that gives the correct answer with certainty.

\begin{definition}
We say that a POVM $\{M_1,...,M_r\}$ \textit{perfectly distinguishes} the states $\varrho_1,...,\varrho_r \in \mathcal{S}(\mathcal{H}) $ if
\[
\Tr \varrho_i M_i =1, \qquad \forall i\in[r],
\]
i.e., the probability of correctly identifying the state of the system is 1 for every possible state. We say that some states $\varrho_1,...,\varrho_r \in \mathcal{S}(\mathcal{H}) $ are \textit{perfectly distinguishable}, if there exists a measurement perfectly distinguishing them.
\end{definition}
 
\begin{proposition}\label{prop:perfect}
A set of states $\varrho_1,...,\varrho_r \in \mathcal{S}(\mathcal{H}) $ are perfectly distinguishable, if and only if they are pairwise orthogonal, i.e.,
\[
\varrho_i \perp \varrho_j, \qquad i\neq j.
\]
\end{proposition}

\begin{remark}
For the definition and equivalent characterizations of orthogonality (in Proposition \ref{prop:ort}) see Subsection \ref{sec:Not}.
\end{remark}

\begin{corollary}
A set of pure states $\ket{\psi_1}\bra{\psi_1},...,\ket{\psi_r}\bra{\psi_r} \in \mathcal{S}(\mathcal{H})  $ are perfectly distinguishable, if and only if $\psi_i \perp \psi_j,\ i\neq j$.
\end{corollary}

\begin{corollary}
The maximum number of perfectly distinguishable states of
a system is the dimension of the Hilbert space of the system. And any such maximal set of perfectly distinguishable states consists of pure states corresponding to an orthonormal basis of $\mathcal{H}$ such that the states are the orthogonal projections to the basis vectors. And the measurement achieving perfect discrimination is uniquely given as the orthogonal projections to the basis vectors.
\end{corollary}

As a consequence of Corollary, the maximum number of perfectly distinguishable states in a quantum system is the same as in the classical system (with the same dimension). There is, however, a difference between classical and quantum systems in that in a classical system there is one unique set of perfectly distinguishable states, while in a quantum system there are continuum many. 

\begin{comment}
Despite the considerably larger state space of a quantum system offers no benefit over a classical system of the same dimension, as long as we are only interested in using these systems for perfect state discrimination. 
Now the reader may wonder do quantum systems offer any benefit over classical ones if we try to use them to transmit more messages than the dimension of the system.
\end{comment}

%_______________________________________________________________ MINIMUM
\subsection{Minimum error state discrimination}

Ideally, Alice and Bob should use perfectly distinguishable states (as we have seen above, they are precisely the states which are orthogonal). However, the state set by Alice may be corrupted ``on the way'', e.g., by some noise in the channel of communication. Here, instead of considering noisy channels, we shall formally regard our channel a perfect one (i.e., it is the identity map) but instead restrict Alice from using arbitrary states for encoding. To put it another way, we shall simply ``prescribe'' to Alice which states she must use. So it makes sense to consider the case when the states are not perfectly distinguishable (that is, not orthogonal), therefore mistake occurs in the decoding inevitably. The aim of \textit{minimum error state discrimination} is to minimize this error, or equivalently, maximize the success. For this we have to define what we call success. A natural way to do that is simply to sum up the $i$-th success probability with the $p_i$ weights (called \textit{prior probabilities}). (We can think $p_i$ as it shows how often Alice sends the message $\varrho_i$ or as the weigh showing how important the $\varrho_i$ message is.)

\begin{align}
   & P_s(\{M_1,...,M_r\}):=\sum_{i=1}^r p_i \Tr \varrho_i M_i  \label{eq:p_s},\\  
   & P_e(\{M_1,...,M_r\}):=\sum_{i=1}^r p_i \Tr \varrho_i (I-M_i)  \label{eq:p_e},
\end{align}
and obviously
\[
P_s(\{M_1,...,M_r\})+P_e(\{M_1,...,M_r\})=1,
\]
so it is enough to deal with one of them. $P_s(\{M_1,...,M_r\})$ is called the \emph{Bayesian success probability}. Which is a function of the $M:=\{M_1,...,M_r\}$ measurement $\povm$. The aim of minimum error state discrimination is to maximize the success probability,

\[
P_s^*(p_1\varrho_1,...,p_r\varrho_r):=\max_{\{M_1,...,M_r\} \povm} \sum_{i=1}^r p_i \Tr \varrho_i M_i.
\]
This is called the \textit{optimal success probability}. We say that a measurement $M$ is \emph{optimal} if it attains the optimal success probability, i.e., $P_s(M)=P_s^*$.

In the following theorems we almost never use the fact that the operators $\varrho_i$ are density operators and the real numbers $p_i$ form a probability distribution. Thus we will replace the operators $p_i \varrho_i$ with arbitrary PSD operators $A_i$  (sometimes called \textit{generalized state}) and consider the following \textit{generalized} problem
\[
P_s^*(A_1,...,A_r):=\max_{\{M_1,...,M_r\} \povm} \sum_{i=1}^r  \Tr A_i M_i,
\]
where $A_1,...,A_r$ are given PSD operators. The above introduced quantities (e.g., the optimal success probability) can be straightforward generalized to this case (by replacing $p_i \varrho_i$ with $A_i$), and we will use the same notation for them.

This generalization has two advantages, on the one hand it simplifies the notation, on the other hand, it gives more general results. %However in some points we want to take advantages of the $A_i$ are not arbitrary, for that we introduced the following notation

\begin{comment}
\begin{definition}
We say that a set of non-zero PSD operators $\{A_1,...,A_r\}$ forms a set of \textit{weighted states} if $A_i=p_i \varrho_i ,\ i\in[r]$ for some $\varrho_i$ states and $p_1,...,p_r$ probability distribution.
\end{definition}
\end{comment}

\begin{remark}
There exists an optimal measurement, since the POVM forms a compact convex set, and $M\mapsto P_s(M)$ is affine (thus continuous). Moreover, there exist an optimal measurement which is extremal.
\end{remark}
%_______________________________________________________________ CL
\subsubsection{Classical case}
First, we consider the classical case (that is, $A_i \in  \mathcal{B}(\mathcal{H})_{+}^\beta$) for some intuition. In this case we can give explicitly the optimal success probability and an optimal measurement. We illustrate it in the following example.
\begin{example}
Consider the following three generalized classical states (in $\mathcal{B}(\mathcal{H})_{+}^\beta$),
\[
A_1=\begin{pmatrix}
3 & &  \\
  &1 &  \\
 & & 2 \\
\end{pmatrix},\quad
A_2=\begin{pmatrix}
1 & &  \\
  &2 &  \\
 & & 2 \\
\end{pmatrix},\quad
A_3=\begin{pmatrix}
2 & &  \\
  &0 &  \\
 & & 1 \\
\end{pmatrix},\quad
\]
then
\[
M_1=\begin{pmatrix}
1 & &  \\
  &0 &  \\
 & & \alpha \\
\end{pmatrix},\quad
M_2=\begin{pmatrix}
0 & &  \\
  &1 &  \\
 & & 1-\alpha \\
\end{pmatrix},\quad
M_3=\begin{pmatrix}
0 & &  \\
  &0 &  \\
 & & 0 \\
\end{pmatrix},\qquad  \alpha\in[0,1],
\]
maximise the $\sum_{i=1}^r  \Tr A_i M_i$ success. That is, roughly speaking, in each ``row'' of the measurement operators we have to put 1 where it is the largest element
 in the same ``row'' of the states. If there are several largest ones, then we can arbitrarily distribute the weights between them. We can see this if we group the sum row-wise, and use the fact that each ``row'' of the POVM sum up to 1. This kind of measurements are called \textit{maximum likelihood measurement}s, and the maximum can be calculated as
\begin{align*}
\sum_{i=1}^3  \Tr A_i M_i 
&=\Tr\begin{pmatrix}
3 & &  \\
  & 2 &  \\
 & & 2  \\
\end{pmatrix} \\
&=\Tr 
\begin{pmatrix}
\max_i\{A_i(1)\} & &  \\
  & \max_i\{A_i(2)\} &  \\
 & & \max_i\{A_i(3)\}  \\
\end{pmatrix} 
=\Tr\, \sup \{A_1,A_2,A_3\},
\end{align*}
where regarding the supremum we consider $\mathcal{B}(\mathcal{H})_{\sa}^\beta$ as a partially ordered set respect to the PSD ordering (it will be explained below in more detail).

We can observe that although the maximising measurement is not unique (freedom of $\alpha$) the $\sup \{A_1,A_2,A_3\}$ is a unique matrix and so the trace of it (which is $7$ for all $\alpha\in[0,1]$). Furthermore if we set $\alpha=0$ or $1$ we can observe that there exist a maximising measurement which is extremal.
\end{example}
After the example we give the formal definition of maximum likelihood measurement.
\begin{definition}
We say that a measurement $\{M_1,...,M_r \}$ is a \textit{maximum likelihood measurement} if for all $\omega$,
\[
A_i(\omega)< \max_j A_j(\omega) \quad \Rightarrow \quad M_i(\omega)=0.
\]
\end{definition}
That is, on any elementary event $\omega$, the measurement only outputs such $i$ values where $A_i(\omega)$ is maximal. 

As we saw, the important feature of the classical case is that we can give explicitly the optimal success probability and an optimal measurement (like in the example).
\begin{proposition}\label{prop:C}
For $A_1,...,A_r\in \mathcal{B}(\mathcal{H})_{+}^\beta$ (i.e., classical case) the optimal success probability is
\begin{equation}\label{eq:Cps}
P_s^*(A_1,...,A_r)=\Tr\, \sup \{A_1,...,A_r\},
\end{equation}
and a measurement M is optimal if and only if it is a maximum likelihood measurement.
\end{proposition}
Here the supremum means that
\begin{equation}\label{eq:supC}
{\sup} \{A_1,...,A_r\}:=\min_Y \{Y :\ Y\geq A_k,\ k=1,...,r  \},
\end{equation}
which exists because $\mathcal{B}(\mathcal{H})_{sa}^\beta$ forms not only a partially ordered set (for the PSD ordering), but a lattice, that is, for all nonempty finite subsets of $\mathcal{B}(\mathcal{H})_{sa}^\beta$ have both a supremum and an infimum.

%_______________________________________________________________ Q
\subsubsection{Quantum case}

One may wonder whether the expression \eqref{eq:Cps} in Proposition \ref{prop:C} can be straightforwardly generalized to the non-commutative case $\mathcal{B}(\mathcal{H})_{sa}$. The answer is “no”, because a finite set of PSD operators in general does not have a supremum. This is because on the RHS of \eqref{eq:supC}, the set $\{Y:Y\geq A_k,\ k=1,...,r  \}$ has minimal element only within $\mathcal{B}(\mathcal{H})_{sa}^\beta$, it has no minimal element within $\mathcal{B}(\mathcal{H})_{sa}$ (not even when the operators $A_k$ mutually commute, e.g., $\diag (1,2)$ and $\diag (2,1)$). However, it has a unique minimal element within $\mathcal{B}(\mathcal{H})_{sa}$ in terms of the \textit{trace ordering}. In view of the above, we can define the following ``hybrid supremum'' (as in \cite{Milan}),
\begin{equation}\label{eq:supQ}
{\sup}^* \{A_1,...,A_r\}:=\underset{Y}{\arg \min} \{\Tr Y :\ Y\geq A_k,\  k=1,...,r  \}. 
\end{equation}
Which is not a supremum in the ordinary mathematical sense, because it uses two different orderings (this is what the word hybrid refers to). It creates the set of upper bounds $\{Y:Y\geq A_k,\ k=1,...,r  \}$ with respect to the PSD ordering, but calculate the minimum with respect to the trace ordering.

With the above introduced ``hybrid supremum'' we can give the optimal success probability (with a formula \eqref{eq:QCps} similar to the classical one \eqref{eq:Cps}) and an optimal measurement. 

\begin{theorem}
For $A_1,...,A_r\in \mathcal{B}(\mathcal{H})_{+}$ (i.e., quantum case) the optimal success probability is
\begin{align}\label{eq:Qps}
P_s^*(A_1,...,A_r)&=\min \{\Tr Y :\ Y\geq A_k,\ k=1,...,r  \}  \\
& =\Tr {\sup}^* \{A_1,...,A_r\} \label{eq:QCps},
\end{align}
and a measurement M is optimal if and only if
\begin{align}\label{eq:max_measure}
\forall i: \quad  A_i\leq \sum_{k=1}^r A_k M_k,
\end{align}
and morover, an optimal measurement satisfies
\[
\forall i: \quad  M_i(A_i-A_k)M_k=0.
\]
\end{theorem}
A major difference between the classical and the quantum case is that in the latter, an optimal measurement can not be explicitly given in general, and \eqref{eq:max_measure} can only be used to verify that a given measurement is optimal or not. Likewise, the optimal success probability can be computed explicitly in the classical case, but not in the general quantum case. An important exception is when there are only two states, which we will discuss in the next subsection.

%_______________________________________________________________ BINARY
\subsection{Binary state discrimination}

The main feature of this special case compared to
the general state discrimination problem is that we can give explicit formulas for the optimal success probability and the measurements achieving it.

\begin{table}[H]
\begin{tabular}{ccll}
& \textit{Alice} & $\varrho_i \ \longrightarrow \ \varrho_i$ & \textit{Bob}    \\[10pt]
$p  \quad $                   & 1: $\varrho_1$                     &                                                           & 1: $M_1:=T$     \\[10pt]
$1-p  \quad $               & 2: $\varrho_2$                     &                                                           & 2: $M_2:=I-T$
\end{tabular}
\end{table}
Now we turn to the general problem as before,
\begin{table}[H]
\begin{tabular}{ccll}
& \textit{Alice} & $A_i \ \longrightarrow \ A_i$ & \textit{Bob}    \\[10pt]
                    & 1: $A_1$                     &                                                           & 1: $T$     \\[10pt]
                & 2: $A_2$                     &                                                           & 2: $I-T$
\end{tabular}
\end{table}
In this case, the error probability defined as
\begin{align}\label{eq:error_binary}
   P_e(T):=\Tr A_1 (I-T)+\Tr A_2 T  
\end{align}
(based on \eqref{eq:p_e}), and the optimal error probability is
\begin{align}\label{eq:error_opt_binary}
   P_e^*(A_1,\, A_2):=\min_{0\leq T\leq I} \Big\{\Tr A_1 (I-T)+ \Tr A_2 T\Big\}.
\end{align}

As we told, the important feature of the binary case is that we can give explicitly the optimal success probability and an optimal measurement.
\begin{theorem}\label{th:P^*}
For $A_1,A_2\in \mathcal{B}(\mathcal{H})_{+}$ (i.e., binary quantum case) the optimal success probability is
\begin{align}\label{eq:P^*}
    P_e^*(A_1,A_2)=\Tr \frac{1}{2}(A_1+A_2-\abs{A_1-A_2}),
\end{align}
and the minimum is attained by an operator 
\[
\{A_1-A_2 >0 \} \leq T \leq \{A_1-A_2 \geq 0 \}.
\]
\end{theorem}
Here we used the following notation. For a self-adjoint operator $X=\sum_\lambda \lambda P(\lambda)$ we denote $\{X >0 \}:=\sum_{\lambda>0} P(\lambda)$, that is the orthogonal projection to the subspace corresponding to the strictly positive eigenvalues (and for $\geq$ similarly).

In the case of states, that is, $A_1=p\varrho_1,\ A_2=(1-p)\varrho_2$, then  $\Tr(A_1+A_2)=1$, so \eqref{eq:P^*} reduces to
\[
P_e^*(A_1,A_2)=\frac{1}{2}-\frac{1}{2} \norm[\big]{A_1-A_2}_1.
\]

Finally we give an upper bound for the error probability in terms of the product of the states.
\begin{proposition}\label{prop:Au}
Let $A,B\in \mathcal{B}(\mathcal{H})_+$, then
\[
\frac{1}{2}(\Tr A +\Tr B)-\frac{1}{2}\norm{A-B}_1 \leq \Tr A^\alpha B^{1-\alpha}, \qquad \forall \alpha \in [0,1].
\]
\end{proposition}
We define the following quantity (which is a positive number), what we will use in asymptotic binary state discrimination, in the next section.
\begin{definition}\label{def:Chernoff}
For $A,B\in \mathcal{B}(\mathcal{H})_+$, we define
\begin{equation}
C(A,B):=-\min_{0\leq \alpha \leq 1} \log \Tr A^\alpha B^{1-\alpha},
\end{equation}
the \emph{Chernoff divergence} of $A$ and $B$.
\end{definition}

%=============================================================== ASYMPTOTIC
\section{Asymptotic state discrimination}\label{sec:Asy}

Consider the communication scenario what we introduced in the previous section. A natural way to reduce the error if Alice sends the message several times (say $n$ times).
\begin{table}[H]
\begin{tabular}{clll}
\multicolumn{1}{l}{} & \multicolumn{1}{l}{\textit{Alice}} & \multicolumn{1}{c}{$\varrho_i^{\otimes n} \ \longrightarrow \ \varrho_i^{\otimes n}$} & \textit{Bob}    \\[10pt]
$p_1 \quad$                   & 1: $\varrho_1$                     &                                                           & 1: $M_1$     \\[10pt]
$p_2 \quad $               & 2: $\varrho_2$                     &                                                           & 2: $M_2$     \\[10pt]
               &   $\vdots$                     &                                                           &   $\vdots$     \\[10pt]
$p_r \quad $               & $r$: $\varrho_r$                     &                                                           & $r$: $M_r$     \\[10pt]
\end{tabular}
\end{table}
It is known that the error vanishes exponentially fast in the number of repetitions $n$.
The aim of \textit{asymptotic state discrimination} is to determine this exponent. In the case of \textit{binary state discrimination}, i.e., when $r=2$ (so there are only two possible messages), the exponent is the well-known \emph{Chernoff divergence} ($C(\varrho_1,\varrho_2)$) of the states $\varrho_1$ and $\varrho_2$, what we will discuss in detail in the next section. (In the case $r>2$, the exponent is the ``worst case exponent'', that is the minimum of the Chernoff divergences of each pairs, i.e., $\min_{i\neq j}C(\varrho_i,\varrho_j)$. This was a long-standing open problem (the Nussbaum and Szkoła's conjecture) which Ke Li solved in 2015 in \cite{KeLi}.)

\subsection{Binary state discrimination}

\begin{table}[H]
\begin{tabular}{ccll}
  &  \textit{Alice}  &  $\varrho_i^{\otimes n} \ \longrightarrow \ \varrho_i^{\otimes n}$  & \textit{Bob}    \\[10pt]
$p $                   & 1: $\varrho$                     &                                                           & 1: $M_1:=T$     \\[10pt]
$q $               & 2: $\sigma$                     &                                                           & 2: $M_2:=I-T$
\end{tabular}
\end{table}
where the $p,q$ are the prior probabilities. We define the following notation for the series of states
\[
\vec \varrho:=\{\varrho^{\otimes n}\}_{n\in \mathbb{N}} \ \textrm{,}\quad  \vec \sigma:=\{\sigma^{\otimes n}\}_{n\in \mathbb{N}},
\] 
 and we use the notation $\varrho_n$ and $\sigma_n$ for $n$-th element of the series. With this notation, the optimal error probability (defined in \eqref{eq:error_opt_binary}) 
\[
P_e^*(p\varrho_n,\,q\sigma_n) =\min_{0\leq T_n\leq I} \Big\{ p\Tr \varrho_n (I-T_n)+ q\Tr \sigma_n T_n \Big\}
\]
is a function of $n$ (where $T_n \in \mathcal{S}(\mathcal{H}^{\otimes n}))$. The error probability is known to decay exponentially fast in the number of repetitions $n$, and hence we are interested in the exponents (which are negative numbers) defined as follows.
\begin{definition}\label{def:p_e}
We define the following \textit{error exponents}
\begin{align*}
  \underline{p}_e(\vec \varrho,\,\vec\sigma) &:=\liminf_{n\to \infty} \frac{1}{n} \log P_e^*(\varrho_n,\,\sigma_n), \\
    \overline{p}_e(\vec \varrho,\,\vec\sigma) &:=\limsup_{n\to \infty} \frac{1}{n} \log  P_e^*(\varrho_n,\,\sigma_n) 
\end{align*}
and
\begin{align*}
  p_e(\vec \varrho,\,\vec\sigma) &:=\lim_{n\to \infty} \frac{1}{n} \log P_e^*(\varrho_n,\,\sigma_n), \\
\end{align*}
if the limit exists.
\end{definition}

\begin{remark}
Recall that for a sequence of real numbers $(x_n)$, the limit superior and limit inferior always exist (as the extended real number line is complete), and  
\[\liminf_{n\to \infty} x_n \leq \limsup_{n\to \infty} x_n.
\]
The ordinary limit exists precisely when equality holds, and then it has the same value.
\end{remark}

\begin{remark}\label{re:without}
We defined the error exponents without consideration of the prior probabilities. This is because
\[
\min\{p,q\}P_e^*(\varrho_n,\,\sigma_n) \leq P_e^*(p\varrho_n,\,q\sigma_n) \leq  \max\{p,q\} P_e^*(\varrho_n,\,\sigma_n),
\]
and taking the limit, it follows that
\[
\liminf_{n\to \infty} \frac{1}{n} \log P_e^*(p\varrho_n,\,q\sigma_n) =\liminf_{n\to \infty} \frac{1}{n} \log P_e^*(\varrho_n,\,\sigma_n),
\]
that is the prior probabilities does not influence these quantities (the same is true for the supremum and for the ordinary limit).
\end{remark}

The following theorem gives the answer for the error exponent, which is the well-known quantum Chernoff bound \cite{Aud1, Aud2,Szkola}. 

\begin{theorem}[quantum Chernoff bound]\label{th:Chernoff}
The limit $p_e(\vec \varrho,\,\vec\sigma)$ exists, and
\begin{equation}\label{eq:chernoff}
 p_e(\vec \varrho,\,\vec\sigma)=-C(\varrho,\sigma),
\end{equation}
where $C(\varrho,\sigma)$ is the \emph{Chernoff divergence} of $\varrho$ and $\sigma$ (see Definition \ref{def:Chernoff}).
\end{theorem}

\begin{remark}
The $\leq$ direction in \eqref{eq:chernoff} is a direct consequence of Proposition \ref{prop:Au}. 
\end{remark}

%_______________________________________________________________ COMP
\subsection{Composite binary state discrimination (the conjecture)}

Suppose that the system what Bob got from Alice is in the state corresponding to one of the messages $i\in\{0,...,r\}$,  but this time he is not interested in the precise value of $i$, he only wants to decide whether $i$ falls in a certain subset $H$. That is, Bob partitions $H$ into $k$ subsets. (So this time the hypotheses are represented by sets of states.) This problem is called \textit{composite hypotheses testing}. When Bob partition $H$ into two subsets (i.e., $k=2$), the problem in called \textit{composite binary hypotheses testing}. We will consider this latter in the simplest case when one of the subsets contains only one element. 
So, for example, Bob may only be interested in that he got the message $0$ or not. That is, he wants to distinguish the two sets  $\{\varrho\}$ and $\{\sigma_1,...,\sigma_r\}$ (as illustrated in the following table).
\begin{table}[H]
\begin{tabular}{cccl}
  &  \textit{Alice}  &   $\quad \varrho_i^{\otimes n} \ \longrightarrow \ \varrho_i^{\otimes n}$  & \hspace{28pt} \textit{Bob}    \\[10pt]
$p_0 $                   & 0: $\varrho$                    &                                                           &\hspace{29pt}  0:  $M_1:=T$     \\[10pt]
$p_1 $                   & 1: $\sigma_1$                    &                                                           & \{1,...,r\}:  $M_2:=I-T$     \\[10pt]
                  &   $\vdots$                     &                                                           &   \\[10pt]
$p_r $               & r: $\sigma_r$                     &                                                           &  
\end{tabular}
\end{table}
\noindent
Although it is known that the error vanish exponentially fast in the number of repetitions, this time the exponent is unknown.
\begin{remark}
In the above setup, Bob has to distinguish between the following two states
\[
 \varrho^{\otimes n},\quad  \sigma_1^{\otimes n}+ \dots + \sigma_r^{\otimes n}.
\]
It is important that Alice does not know what kind of grouping Bob is interested in, that is, how he partitioned the outcomes. If Alice knew it in advance, of course, she would send just (the information from which set the state is, that is) the convex combination of the states $n$ times, and in this case Bob should distinguish between
\[
 \varrho^{\otimes n},\quad  (\sigma_1+ \dots + \sigma_r)^{\otimes n},
\]
which is the simpler, non-composite binary state discrimination problem (in the previous subsection). Obviously, in the single shot case (i.e., $n=1$) these two problems are equivalent.
\end{remark}
Like before, we will immediately consider the \textit{generalized} problem, so instead of the states with the prior probabilities, we will consider arbitrary PSD operators.
\begin{table}[H]
\begin{tabular}{cccl}
  &  \textit{Alice}  &   $\quad \varrho_i^{\otimes n} \ \longrightarrow \ \varrho_i^{\otimes n}$  & \hspace{28pt} \textit{Bob}    \\[10pt]
                   & 0: $A$                    &                                                           &\hspace{29pt}  0:  $M_1:=T$     \\[10pt]
                   & 1: $B_1$                    &                                                           & \{1,...,r\}:  $M_2:=I-T$     \\[10pt]
                  &   $\vdots$                     &                                                           &   \\[10pt]
                & r: $B_r$                     &                                                           &  
\end{tabular}
\end{table}
In this scenario we define the \textit{worst-case error probability} as 
\begin{align}
   P_e(T):=\Tr A (I-T)+\max_i\{\Tr B_i T \} 
\end{align}
(similarly like \eqref{eq:error_binary}), and the \textit{optimal worst-case error probability} as 
\begin{align}
   P_e^*( A,\, \{B_1,...,B_r \}):=\inf_{0\leq T\leq I} \Big\{\Tr A (I-T)+\max_i\{\Tr B_i T \}\Big\}.
\end{align}
By the same argument as in Reference \ref{re:without}, we define the following composite error exponents without consideration of the prior probabilities. (We denote the sequences of states $\vec \varrho:=\{\varrho^{\otimes n}\}_{n\in \mathbb{N}}$ and
 $\vec \sigma:=\{\sigma^{\otimes n}\}_{n\in \mathbb{N}}$ as before, where $\varrho_n$ and $\sigma_n$ stands for the $n$-th element of the series.) 
\begin{definition}
We define the following \textit{composite error exponents}
\begin{align*}
  \underline{p}_e(\vec \varrho, \{\vec\sigma_{1},...,\vec\sigma_{r}\}) &:=\liminf_{n\to \infty} \frac{1}{n} \log P_e^*(\varrho_n,\{\sigma_{1,n},...,\sigma_{r,n}\}), \\
    \overline{p}_e(\vec \varrho, \{\vec\sigma_{1},...,\vec\sigma_{r}\}) &:=\limsup_{n\to \infty} \frac{1}{n} \log P_e^*(\varrho_n,\{\sigma_{1,n},..., \sigma_{r,n}\})
\end{align*}
and
\[
p_e(\vec \varrho, \{\vec\sigma_1,...,\vec\sigma_{r}\}) :=\lim_{n\to \infty} \frac{1}{n} \log P_e^*(\varrho_n,\{\sigma_{1,n},...,\sigma_{r,n}\}),
\]
if the limit exists.
\end{definition}
As we said, the error exponents are unknown, nevertheless we have trivial lower and upper bounds because of the following. For every $j$ and $T$, we have
\[
\Tr A (I-T)+\Tr B_j T \leq \Tr A (I-T)+\max_i\{\Tr B_i T \}  \leq  \Tr A (I-T)+\sum_i\Tr B_i T,
\]
and taking the infimum in $T$ yields the following trivial single-shot inequality
\begin{equation}
\max_i P_e^*(A,B_i) \leq P_e^*(A,\{B_1,\,...,B_r\}) \leq   P_e^*(A,\sum_i B_i).
\end{equation}
This is true for all $A=\varrho_n$ and $B=\sigma_{i,n}$ states, therefore by, taking the limit in $n$, we get the following proposition.
\begin{proposition}\label{pr:compineq}
For the composite error exponents, we have the following lower and upper bounds
\[
\max_i  p_e(\vec \varrho,\,\vec\sigma_i) \leq  \underline{p}_e(\vec \varrho, \{\vec\sigma_{1},..., \vec\sigma_{r}\}) \leq \overline{p}_e(\vec \varrho, \{\vec\sigma_{1},...,\vec\sigma_{r}\})  \leq  \overline{p}_e(\vec \varrho,\,\sum_{i=1}^r \vec\sigma_i).
\]
\end{proposition}
The following conjecture says that the above lower bound for the exponent is precisely the exponent (which is unsolved for several years, see in \cite{Milan}). 
\begin{conjecture}\label{conj} 
The limit $p_e(\vec \varrho, \{\vec\sigma_{1},..., \vec\sigma_{r}\})$ exists, and
\begin{equation}\label{eq:conj}
p_e(\vec \varrho, \{\vec\sigma_{1},..., \vec\sigma_{r}\})=\max_i  p_e(\vec \varrho,\,\vec\sigma_i),
\end{equation}
where $\vec \varrho:=\{\varrho^{\otimes n}\}_{n\in \mathbb{N}}$ and
 $\vec \sigma_i:=\{\sigma_i^{\otimes n}\}_{n\in \mathbb{N}}$ for $i\in[r]$.
\end{conjecture}
This conjecture is open even in the simplest case $r=2$, that is, when the first set consists of only one state $\{\varrho\}$ and the second consists of only two states $\{\sigma_1,\sigma_2\}$. Nevertheless, there are some known special cases (presented in the next section), where the conjecture is known to be true.

\begin{remark}
The $\geq$ direction in \eqref{eq:conj} is trivial from Proposition \ref{pr:compineq}. The interesting unknown part is the $\leq$ direction.
\end{remark}

%=============================================================== 
\subsection{Known special cases}
The composite binary state discrimination problem (introduced in the previous section) is specified by two sets of states. In this subsection we introduce the known special cases (by the special sets of states). In these cases, analytical proofs show that the conjecture is true there. The goal is here to compare the cases in question without any detailed explanation. The following table serves as an overview of these cases. (The conjecture was stated for the second row.)

\begin{align*}
    \textrm{unknown case: } \quad    &\{\varrho_1,...,\varrho_r\}  \\
    &\{\sigma_1,...\sigma_m\}  \\[15pt]
    \textrm{simpler unknown case: } \quad    &\{\varrho\}  \\
    &\{\sigma_1,...\sigma_m\}  \\[15pt]
    \textrm{known special cases:\quad  case 1:} \quad    &\{\varrho\}  \\
    &\{\sigma\}  \\[15pt]
    \textrm{(classical) case 2:} \quad    &\{\Large(\ddots ),(\ddots )...(\ddots ) \} \\
    &\{\Large(\ddots ),(\ddots )...(\ddots ) \}  \\[15pt]
    \textrm{case 3:} \quad    &\{\ke{\varphi} \br{\varphi}\}  \\
    &\{\varrho_1,...\varrho_r\}  \\[15pt]
    \textrm{case 4:} \quad    &\{\ke{\varphi_1} \br{\varphi_1},...,\ke{\varphi_r} \br{\varphi_r}\}  \\
    &\{\ke{\psi_1} \br{\psi_1},...,\ke{\psi_m} \br{\psi_m}\}  \\[15pt]
    \textrm{case 5:} \quad    &\{\varrho\}  \\
    &\{\Large(\ddots ),(\ddots ),...,(\ddots ) \}  \\[15pt]
    \textrm{My work:  case 6:} \quad    &\Big\{\frac{P}{\Tr P}\Big\}  \\
    &\Big\{\frac{Q}{\Tr Q}\, ,\ \ke{\varphi} \br{\varphi}\Big\}  
\end{align*}
In the table, $\varrho_i, \sigma_i$ denote arbitrary states (from the same space). Within a case, ($\ddots$) denotes commuting states (emphasize they are all diagonal in the same basis). The pure states are denoted by $\ke{\varphi} \br{\varphi}$ as usual, and $P,\,Q$ stand for orthogonal projections.

\begin{remark}
The general unknown case (first row of the table) in the special setup, when the worst distinguishable pair is not from the same set, the conjecture is true because of \cite{KeLi}.
\end{remark}

Case 6 is my work, which is the subject of the next section, where we will introduce it in detail. At this point, we should only observe the following. All of the known special cases are in some sense ``nice''; e.g., having certain symmetries (for example, there is some underlying group structure). However case 6 is more ``asymmetrical'' in some sense. The reason why it is exciting regarding the conjecture, is revealed in the next section.

%=============================================================== 
\subsection{Verifying a new special case (my work)}\label{sec:Ver}

%________________________________________________________________THEOREM
In all of the above cases, analytical proofs show that the conjecture is true there. Numerical computations done in low dimensional cases also suggest that the conjecture is true. However, all of the known special cases are in some sense ``nice'', e.g., having certain symmetries. As we have mentioned earlier, this may make the statement of the conjecture true for them, for example, because of some underlying group structure. Numerical computations can also be misleading: the deviation from the conjectured value could be very small, and to get high precision asymptotic rates, one needs to consider high tensorial powers, which means that even if we start in low dimensions, one needs computations with extremely large matrices. Therefore, it could happen that in a less ``nice'' case the conjecture fails. This was precisely what happened in the well-known question of ``\textit{Superadditivity of communication capacity using entangled inputs}'' \cite{Hastings2009}, where a certain  natural conjecture was made, and for a long time many great researchers were attempting to prove it in vain. Numerical searches did not find counterexamples, and in all special cases that could be easily considered, the conjecture was true. However, in the end, the conjecture was shown to be false by a random counterexample. So the reason people could not find a counterexample was because the cases that are easy to deal with are just ``too nice''.   

Because of the above, one is motivated to find 
cases which on the one hand are as ``asymmetrical'' as possible (and, of course, do not follow from any of the known cases), yet still analytically computable. Following a suggestion of my supervisor, I will consider the two sets of states $\{\varrho\}$ and $\{\sigma_1,\sigma_2\}$ with the following assumptions:
\begin{itemize}
    \item[(1)] $\varrho$ and $\sigma_1$ are multiples of commuting projections,
    \item[(2)] $\sigma_2$ is pure,
    \item[(3)] none of these projections is contained in any of the other ones,
    \item[(4)] $\sigma_2$ does not commute with any of the other ones,
    \item[(5)] neither $\varrho$, nor $\sigma_1$ are pure.
\end{itemize}
\begin{remark}\label{re:ass}
Assumptions (1) and (2) are done to make the case computable, while (3), (4) and (5) serve to make it not fall in any of the known cases, and render it as ``asymmetrical'' as possible.
\end{remark}
As it turns out from the following theorem, the conjecture is also true in this new special case, which increases the likelihood that the conjecture is true in general.

\enlargethispage{-\baselineskip}
\begin{theorem}\label{th:mywork} 
Let $\varrho, \sigma_1$ and $\sigma_2$ as explained before. Then the limit $p_e(\vec \varrho, \{\vec\sigma_{1},\vec\sigma_{2}\})$ exists and  
\[
p_e(\vec \varrho, \{\vec\sigma_{1}, \vec\sigma_{2}\})=\max_i  p_e(\vec \varrho,\,\vec\sigma_i),
\]
where $\vec \varrho:=\{\varrho^{\otimes n}\}_{n\in \mathbb{N}}$ and
 $\vec \sigma_i:=\{\sigma_i^{\otimes n}\}_{n\in \mathbb{N}}$ for $i\in[r]$.
\end{theorem}
In the rest of the section we prove the above theorem.

\begin{remark}\label{re:th}
By Proposition \ref{pr:compineq}, the following is true (for any kind of states)
\[
\max_i  p_e(\vec \varrho,\,\vec\sigma_i) \leq  \underline{p}_e(\vec \varrho, \{\vec\sigma_{1},\, \vec\sigma_{2}\}) \leq \overline{p}_e(\vec \varrho, \{\vec\sigma_{1},\,\vec\sigma_{2}\})  \leq  \overline{p}_e(\vec \varrho,\, \vec\sigma_1+\vec\sigma_2).
\]
So, to prove the theorem, it is enough to show that the following is true (for the states explained before),
\begin{equation}\label{eq:enough}
\overline{p}_e(\vec \varrho,\, \vec\sigma_1+\vec\sigma_2)  = \max_i  p_e(\vec \varrho,\,\vec\sigma_i).
\end{equation}
\end{remark}

We will go through Propositions \ref{prop:RHS}-\ref{prop:B}, and after that we will give the proof of Theorem \ref{th:mywork}. First we introduce some notations for the states defined above, which we will use in the whole paper. For the two commuting states $\varrho$ and $\sigma_1$,

\begin{align*} 
\varrho:&=pP, \quad &p:=\frac{1}{\Tr P}, &&\\ 
\sigma_1:&=qQ,  \quad &q:=\frac{1}{\Tr Q},  &&\\
& &R:=\Tr PQ, 
\end{align*}
where $P$ and $Q$ are two commuting orthogonal projections (i.e., $P=P^2=P^*,\  Q=Q^2=Q^*$ and $PQ=QP$). For the pure state $\sigma_2$,

\begin{align*} 
&\sigma_2:=\ke{\psi} \br{\psi},  \quad &t&:=\inner{\psi, P \psi},  &&\\ 
& &s&:=\inner{\psi, Q \psi},    \\
& &r&:=\inner{\psi, PQ \psi},    
\end{align*}
where $\norm{\psi}=1$. Considering the assumptions (3), (4) and (5), we see that these parameters could be $p,q\in \{\frac{1}{n} : n=2,3,... \} \subset (0,1)$ and $R\in \{0,1,2,...\}$ and $t,s,r\in[0,1)$ and satisfying
\begin{align*}
1-t-s+r &\geq 0 \\
  t-r &\geq 0 \\
  s-r &\geq 0 
\end{align*}
This six parameters do not determine the three states, but however, (we will see) they are determine the optimal error probability, so the error exponent.

%________________________________________________________________PROP RHS
In the following proposition we calculate the RHS of \eqref{eq:enough} in terms of the parameters $p,q,t,s,r,R$.
\begin{proposition}\label{prop:RHS} 
Let $\rho, \sigma_1, \sigma_2$ be the states given as above, then 
\begin{equation}\label{eq:min}
    \max_i  p_e(\vec \varrho,\,\vec\sigma_i)
    =\log\Big(\max\Big\{R\min\{p,q\},pt\Big\}\Big)
\end{equation}
\end{proposition}

\begin{proof}
By Theorem \ref{th:Chernoff} (and Definition \ref{th:Chernoff}), we can calculate the two error exponents $p_e(\vec \varrho,\,\vec\sigma_1)$ and $p_e(\vec \varrho,\,\vec\sigma_2)$ as
\begin{align*} 
p_e(\vec \varrho,\,\vec\sigma_1)
&=\min_{\alpha \in [0,1]} \log{\Tr\Big((pP)^\alpha(qQ)^{1-\alpha}\Big)}= \\
&=\log \Big(\Tr(PQ) \min_{\alpha \in [0,1]} p^\alpha q^{1-\alpha}\Big)
=\log \Big(R\min\{p,q\}\Big)
\end{align*}
and
\begin{align*} 
p_e(\vec \varrho,\,\vec\sigma_2)
&=\min_{\alpha \in [0,1]} \log{\Tr\Big((pP)^\alpha \ke{\psi} \br{\psi}^{1-\alpha}\Big)}= \\
&=\min_{\alpha \in [0,1]} \log \inner{\psi,(pP)^\alpha \psi }
=\log \Big(\inner{\psi,P \psi } \min_{\alpha \in [0,1]} p^\alpha \Big)
=\log (pt)
\end{align*}
and, since $\log$ is a monotone function, the Proposition follows.
\end{proof}

In the following proposition we calculate the LHS of \eqref{eq:enough} in terms of the parameters $p,q,t,s,r,R$.
%________________________________________________________________PROP 2
\begin{proposition}\label{prop:LHS}  
In the case $r=0$,
\begin{equation}\label{eq:LHS_A}
    \overline{p}_e(\vec \varrho,\, \vec\sigma_1+\vec\sigma_2) 
    = \limsup_{n \to \infty} \frac{1}{n} \log \Bigg(2\parr[\Big]{R\min\{p,q\}}^n+1+p^n+q^n 
    - \norm[\big]{A(n)}_1   \Bigg),
\end{equation}
and in the case $r\neq 0$,
\begin{equation} \label{eq:LHS_B}
    \overline{p}_e(\vec \varrho,\, \vec\sigma_1+\vec\sigma_2) 
    = \limsup_{n \to \infty} \frac{1}{n} \log \Bigg(2\parr[\Big]{R\min\{p,q\}}^n+1+p^n+q^n
    +\abs{p^n-q^n} - \norm[\big]{B(n)}_1   \Bigg),
\end{equation}
where
\begin{equation}\label{eq:A}
    A(n)
    =\begin{pmatrix}
    t^n-p^n & \sqrt{(ts)^n} & \sqrt{t^n (1-t^n-s^n)}  \\
    \sqrt{(ts)^n}  & s^n+q^n & \sqrt{s^n (1-t^n-s^n)}  \\
     \sqrt{t^n (1-t^n-s^n)}  & \sqrt{s^n (1-t^n-s^n)}  &  1-t^n-s^n \\
    \end{pmatrix}
\end{equation}
and
\begin{equation}\label{eq:B}
    B(n)
    =\begin{pmatrix}
    t^n-r^n-p^n & \sqrt{(t^n-r^n)r^n} &  \sqrt{(t^n-r^n)(s^n-r^n)} & \sqrt{(t^n-r^n)(1-t^n-s^n+r^n)}   \\
        & r^n+q^n-p^n & \sqrt{r^n(s^n-r^n)} &  \sqrt{r^n(1-t^n-s^n+r^n)}  \\
       &   & s^n-r^n+q^n &  \sqrt{(s^n-r^n)(1-t^n-s^n+r^n)}    \\
         &     &    &  1-t^n-s^n+r^n     \\
    \end{pmatrix}, 
\end{equation}
where $B(n)$ is also a symmetric matrix (but we only write down the upper triangular part, due to lack of space).
\end{proposition}

%______________________________________________________proof
\begin{proof}
From Definition \ref{def:p_e} and Theorem \ref{th:P^*}, it follows that
\begin{align}\label{eq:Ap}
\overline{p}_e(\vec \varrho,\, \vec\sigma_1+\vec\sigma_2) 
 &=\limsup_{n\to \infty} \frac{1}{n} \log P_e^*(\varrho_n,\,\sigma_{1,n}+\sigma_{2,n})= \nonumber\\
 &=\limsup_{n \to \infty} \frac{1}{n} \log \Bigg(3-\norm[\Bigg]{\varrho^{\otimes n}-(\sigma_1^{\otimes n}+\sigma_2^{\otimes n}) }_1 \Bigg)
\end{align}
In the last step the $\frac{1}{2}$ factor (from \eqref{eq:P^*}) disappears by the limit (cf. Remark \ref{re:without}).

In the following we will choose a convenient basis in which the 1-norm is computable, that is, the matrices of the operators have the simplest form. Observe that the commuting projections $P$ and $Q$ partition the $\mathcal{H}$ Hilbert space into four orthogonal subspaces
\begin{align*}
    &\Image \parr[\Big]{P -PQ}, \\
    &\Image PQ, \\
    &\Image \parr[\Big]{Q -PQ}, \\
    &\Image \parr[\Big]{I-P -Q +PQ}.
\end{align*}
The vector $\psi$ can be uniquely decomposed into the sum of vectors from the four subspaces, and the components are proportional to $\sqrt{t-r}$, $\sqrt{r}$, $\sqrt{s-r}$ and $\sqrt{1-t-s+r}$.
Since the tensor products of the commuting projections $P$ and $Q$ are also commuting projections $P^{\otimes n}$ and $Q^{\otimes n}$, they partition the $\mathcal{H}^{\otimes n}$ tensor product space again into four orthogonal subspaces
\begin{align*}
    V_P:&=\Image \parr[\Big]{P^{\otimes n}-(PQ)^{\otimes n}},\\
    V_{PQ}:&=\Image (PQ)^{\otimes n},\\
    V_Q:&=\Image \parr[\Big]{Q^{\otimes n}-(PQ)^{\otimes n}},\\
    V:&=\Image \parr[\Big]{I-P^{\otimes n}-Q^{\otimes n}+(PQ)^{\otimes n}}.
\end{align*}
The vector $\psi^{\otimes n}$ can be uniquely decomposed into the sum of vectors from the four subspaces, and the components are proportional to $\sqrt{t^n-r^n}, \sqrt{r^n}, \sqrt{s^n-r^n}$ and $\sqrt{1-t^n-s^n+r^n}$.
\begin{remark}
We could observe that the tensor product preserves the relationship of these states, so (we will see that) the 1-norm of the $n$-th tensor power can be calculated relatively easily. This is the reason why we state Theorem \ref{th:mywork} for such states.
\end{remark}
We will choose such a basis in which the matrices of the two projections are diagonal, and the matrix of the pure state is as simple as possible. We will consider the $r=0$ and $r\neq0$ case separately. 
\begin{remark}
Since the two orthogonal projections $P$ and $Q$ are commuting, PQ is an orthogonal projection too, so positive, therefore $\inner{\psi, PQ \psi}=0$ (that is the $r=0$ case) is equivalent to $PQ\psi=0$.
\end{remark}

\emph{Case} $r=0$: (that is $PQ\psi=0$) In this case $\psi^{\otimes n}$ could have components only in $V_P, V_Q$ and $V$, so for a fixed $n$, we choose the orthonormal basis $\beta:=\{b_1,b_2,b_3,...\}$ in $\mathcal{H}^{\otimes n}$ in the following way. 
Roughly speaking, we choose $b_1,b_2,b_3$ being parallel to the components of $\psi^{\otimes n}$ in $V_P, V_Q$ and $V$ (which are at least one dimensional, because of assumptions (3), (4) and (5)). In more detail, we choose $b_1\in V_P$ such that $\braket{b_1,\, \psi^{\otimes n}}=\sqrt{t^n}$.
(If $\psi^{\otimes n}$ has nonzero component in $V_P$, then $b_1$ is proportional to it, if the component is the null vector (i.e., $t=0$) then $b_1\in V_P$ is an arbitrary unit vector.) We choose $b_2$ and $b_3$ in the similar way. Then we complete these three vectors to an orthonormal basis

\begin{align*} 
&b_1\ \in V_P \quad \textrm{such that}\ \braket{b_1,\, \psi^{\otimes n}}=\sqrt{t^n},  \\ 
&b_2\ \in V_Q \quad \textrm{such that}\ \braket{b_2,\, \psi^{\otimes n}}=\sqrt{s^n},  \\ 
&b_3\ \in V \quad \textrm{such that}\ \braket{b_3,\, \psi^{\otimes n}}=\sqrt{1-t^n-s^n},  \\ 
&\vdots   \\
&b_i\ \in V_P,  \\ 
&\vdots   \\ 
&b_j\ \in V_{PQ},   \\ 
&\vdots   \\ 
&b_k\ \in V_Q,  \\ 
&\vdots    \\
&b_l\ \in V,  \\ 
&\vdots  
\end{align*}
where the order of the subspaces is important. In this basis, $\psi^{\otimes n}$ has the following coordinate vector 
\begin{align*} 
[\psi^{\otimes n}]_{\beta}
=\begin{pmatrix}
\sqrt{t^n} \\
\sqrt{s^n} \\
\sqrt{1-t^n-s^n} \\ 
0 \\
\vdots   \\ 
0 \\
\end{pmatrix},
\end{align*}
and the three tensor product states have the following matrices (for the easier readability we will use the same notation for the operator and its matrix)
\begin{align*} 
&\varrho^{\otimes n}=p^nP^{\otimes n}=p^n\diag(1,0,0,\overbrace{1,......,\underbrace{1,...,1}_{R^n}}^{1/p^n-1},0,...,0,0,...,0) &&\\ 
&\sigma_1^{\otimes n}=q^nQ^{\otimes n}=q^n\diag(0,1,0,0,...,0,\underbrace{1,...,1,1,...,1}_{1/q^n-1},0,...,0) &&\\ 
\end{align*}
\begin{align*} 
&\sigma_2^{\otimes n}=\ke{\psi^{\otimes n}} \br{\psi^{\otimes n}} 
=\begin{pmatrix}
t^n & \sqrt{t^n s^n} & \sqrt{t^n (1-t^n-s^n)} & \\
\sqrt{t^n s^n}  & s^n & \sqrt{s^n (1-t^n-s^n)} & \\
 \sqrt{t^n (1-t^n-s^n)}  & \sqrt{s^n (1-t^n-s^n)}  &  1-t^n-s^n & \\
    &    &   & 0 \\
   &    &   & & \ddots \\
      &    &   & &  & 0 \\
\end{pmatrix}. &&\\ 
\end{align*}
So the only non-diagonal part is the 3-by-3 submatrix in the top left corner. In this basis we can easily calculate the 1-norm
 \begin{align*} 
\norm[\Bigg]{\varrho^{\otimes n}-(\sigma_1^{\otimes n}+\sigma_2^{\otimes n}) }_1 &= p^n\parr[\Big]{\frac{1}{p^n}-1-R^n}+q^n\parr[\Big]{\frac{1}{q^n}-1-R^n}+\abs[\big]{p^n-q^n}R^n+\norm[\big]{A(n)}_1  \\
&=2-p^n-q^n+R^n\parr[\big]{-p^n-q^n+\abs[\big]{p^n-q^n}}+\norm[\big]{A(n)}_1, 
\end{align*}
where $A(n)$ is the 3-by-3 matrix in \eqref{eq:A}. Now we rewrite $-p^n-q^n+\abs[\big]{p^n-q^n}=-2(\min\{p,q\})^n$, and continue the equation \eqref{eq:Ap}, we get 
\begin{equation}
     \overline{p}_e(\vec \varrho,\, \vec\sigma_1+\vec\sigma_2) 
    =\limsup_{n \to \infty} \frac{1}{n} \log \Bigg(1+p^n+q^n 
    +2\parr[\Big]{R\min\{p,q\}}^n 
    - \norm[\big]{A(n)}_1   \Bigg),
\end{equation}
which is equal to \eqref{eq:LHS_A}, so we finished the proof in the case $r=0$.

\emph{Case} $r\neq0$: (that is, $PQ\psi\neq0$) The proof is the same as in the previous case, the only difference is that this time $\psi^{\otimes n}$ has component in $V_{PQ}$. So the same procedure as before results the following basis
\begin{align*} 
&b_1\ \in V_P \quad \textrm{such that}\ \braket{b_1,\, \psi^{\otimes n}}=\sqrt{t^n-r^n},  \\ 
&b_2\ \in V_{PQ} \quad \textrm{such that}\ \braket{b_2,\, \psi^{\otimes n}}=\sqrt{r^n},  \\ 
&b_2\ \in V_Q \quad \textrm{such that}\ \braket{b_2,\, \psi^{\otimes n}}=\sqrt{s^n-r^n},  \\ 
&b_3\ \in V \quad \textrm{such that}\ \braket{b_3,\, \psi^{\otimes n}}=\sqrt{1-t^n-s^n+r^n},  \\ 
&\vdots   \\
&b_i\ \in V_P,  \\ 
&\vdots   \\ 
&b_j\ \in V_{PQ},   \\ 
&\vdots   \\ 
&b_k\ \in V_Q,  \\ 
&\vdots    \\
&b_l\ \in V,  \\ 
&\vdots  
\end{align*}
where the order of the subspaces is important. In this basi, $\psi^{\otimes n}$ has the following coordinate vector 
\begin{align*} 
[\psi^{\otimes n}]_{\beta}
=\begin{pmatrix}
\sqrt{t^n-r^n} \\
\sqrt{r^n} \\
\sqrt{s^n-r^n} \\
\sqrt{1-t^n-s^n+r^n} \\ 
0 \\
\vdots   \\ 
0 \\
\end{pmatrix},
\end{align*}
and the three tensor product states have the following matrices (for the easier readability we will use the same notion for the operator and its matrix)
\begin{align*} 
&\varrho^{\otimes n}=p^nP^{\otimes n}=p^n\diag(1,1,0,0,\overbrace{1,...,1,\underbrace{1,...,1}_{R^n-1}}^{1/p^n-2},0,...,0,0,...,0), &&\\ 
&\sigma_1^{\otimes n}=q^nQ^{\otimes n}=q^n\diag(0,1,1,0,0,...,0,\underbrace{1,...,1,1,...,1}_{1/q^n-2},0,...,0), &&\\ 
\end{align*}
\begin{align*} 
&\sigma_2^{\otimes n}=\ke{\psi^{\otimes n}} \br{\psi^{\otimes n}}=   \\
=&\begin{pmatrix}
    t^n-r^n-p^n & \sqrt{(t^n-r^n)r^n} &  \sqrt{(t^n-r^n)(s^n-r^n)} & \sqrt{(t^n-r^n)(1-t^n-s^n+r^n)}  & & \\
        & r^n+q^n-p^n & \sqrt{r^n(s^n-r^n)} &  \sqrt{r^n(1-t^n-s^n+r^n)}  & & &\\
       &   & s^n-r^n+q^n &  \sqrt{(s^n-r^n)(1-t^n-s^n+r^n)}  & & & \\
         &     &    &    1-t^n-s^n+r^n & & & \\
         &     &    &     & 0 & & \\
         &     &    &     & & \ddots & \\
         &     &    &     & &   & 0\\
    \end{pmatrix}, 
\end{align*}
where the latter is a symmetric matrix (and we only write down the upper triangular part due to lack of space). So the only non-diagonal part is the 4-by-4 submatrix in the top left corner. In this basis we can easily calculate the 1-norm.
 \begin{align*} 
\norm[\Bigg]{\varrho^{\otimes n}-(\sigma_1^{\otimes n}+\sigma_2^{\otimes n}) }_1 &= p^n\parr[\Big]{\frac{1}{p^n}-1-R^n}+q^n\parr[\Big]{\frac{1}{q^n}-1-R^n}+\abs[\big]{p^n-q^n}(R^n-1)+\norm[\big]{B(n)}_1  \\
&=2-p^n-q^n+R^n\parr[\big]{-p^n-q^n+\abs[\big]{p^n-q^n}}-\abs[\big]{p^n-q^n}+\norm[\big]{B(n)}_1  
\end{align*}
where the $B(n)$ is the 4-by-4 matrix in \eqref{eq:B}. Now we rewrite $-p^n-q^n+\abs[\big]{p^n-q^n}=-2(\min\{p,q\})^n$, and continue the equation \eqref{eq:Ap}, we get 
\begin{equation}
     p_e(\vec \varrho,\, \vec\sigma_1+\vec\sigma_2) 
    =\limsup_{n \to \infty} \frac{1}{n} \log \Bigg(1+p^n+q^n+\abs[\big]{p^n-q^n} 
    +2\parr[\Big]{R\min\{p,q\}}^n 
    - \norm[\big]{B(n)}_1   \Bigg),
\end{equation}
which is equal to \eqref{eq:LHS_B}, so we finished the proof in the case $r\neq 0$. So we finished the proof of Proposition \ref{prop:LHS}.
\end{proof}

%________________________________________________________________PROP 3
In the following two proposition we calculate the 1-norm of $A(n)$ and $B(n)$ (which appear in Proposition \ref{prop:LHS}) in terms of the parameters $p,q,t,s,r$.
\begin{proposition}\label{prop:A} 
Let $A(n)$ be the matrix given in Proposition \ref{prop:LHS}, then
\begin{equation}\label{eq:3norm}
\norm{A(n)}_1=1+p^n+q^n-2(pt)^n+o((pt)^n).
\end{equation}
\end{proposition}

%______________________________________________________proof
\begin{proof} 
\begin{remark}\label{re:notation}
In this proof for the easier readability we will use the slight abuse of notation 
\begin{equation}
   p:=p(n):=p^n ,
\end{equation}
and the same for $q,t, s$. For example, when we write $o(p)$, it means $o(p^n)$ as $n \to \infty$. In this notation $p<q$ means that $p^n<q^n$ (i.e., $p^n=o(q^n)$) which nicely coincides with the ordinary meaning (that is, $p<q$ as numbers). Moreover we will use the convention in Remark \ref{re:ordo}.
\end{remark}

First we consider the case when the parameters $t,s\in(0,1)$. (At the end of the proof we will consider the extreme cases when $t=0$ or $s=0$.)
The 1-norm of $A(n)$ is the sum of the absolute value of the eigenvalues.
Since we know the sum of the eigenvalues ($\Tr A(n)$) it is enough to find out the negative ones to calculate the 1-norm (for the easier readability we do not denote the $n$ dependence of the eigenvalues, i.e., $\lambda_i:=\lambda_i(n)$). We start with the following trivial formulas
 \begin{align} 
 \lambda_1+\lambda_2+\lambda_3&=\Tr (A(n)) =1+q-p >0  \label{sum3}\\ 
 \lambda_1\lambda_2\lambda_3&=\det (A(n)) =-(p q)(1-t-s) <0.     \label{prod3}
\end{align}
It follows that $A$ has two positive $\lambda_2,\lambda_3$ and one negative $\lambda_{1}$ eigenvalue, so
 \begin{align}\label{eq:lambdaminus}
\norm{A(n)}_1&=\abs{\lambda_2}+\abs{\lambda_3}+\abs{\lambda_{1}}= && \nonumber \\ 
 &=\lambda_2+\lambda_3-\lambda_{1}=\Tr (A(n)) -2\lambda_{1}=1+q-p-2\lambda_{1} ,&& 
\end{align}
thus, in the following the goal is to find out $\lambda_-$ (in $o(pt)$ precision). The eigenvalues of $A(n)$ are the roots of its characteristic polynomial,
\begin{align}  
f_n(\lambda):&=-\det(A(n)-\lambda I),
\end{align}
which gives the following cubic equation
\begin{align}   \label{eq:char3}
\lambda^3-(1+q-p)\lambda^2+(-p+q-pq-qs+pt)\lambda+(pq)(1-t-s)=0. 
\end{align}
Of course, we could solve the cubic equation exactly, but this would be a much more difficult way. Instead of that, we take advantage of that we only need the eigenvalue in the order of $o(pt)$ precision. So in the following we are working to get know $\lambda_-$ in the order of $o(pt)$ precision. Observe that $f_n(\lambda) \rightarrow \lambda^2 (\lambda-1)$, therefore
\begin{align*} 
&\lambda_1 \to 0,  \\ 
&\lambda_2 \to 0,   \\ 
&\lambda_3 \to 1,   
\end{align*}
so $\lambda_3=1+o(1)$. 

In the first step, we will improve the order of this approximation. The main idea here is that $1$ is ``almost'' a root of the cubic equation \eqref{eq:char3}, thus it can ``almost'' be factorising by $(\lambda_3-1)$ aside from the term on the RHS (remainder), which is ``very'' fast decreasing, 
\[
(\lambda_3-1)\Big(\lambda_3^2+(p-q)\lambda_3 - pq-qs+pt\Big)=+qs(1+p)+pt(-1+q).
\]
We can write $\lambda_3=1+\Delta$ where $\Delta=o(1)$, with this
\[
\Delta\Big((1+\Delta)^2+(p-q)(1+\Delta) -pq-qs+pt\Big)=o(p+q)
\]
Since $\Delta=o(1)$, the second factor on the LHS is $\sim 1$, it follows that $\Delta=o(p+q)$ so we get that
\begin{equation}
\lambda_3=1+o(p+q).
\end{equation}
We will use the above procedure several times and will refer to it as the ``factorising trick''. 

Now we turn our focus to the other two eigenvalues. Using the above result about $\lambda_3$ by \eqref{sum3} and \eqref{prod3}, it follows that
\begin{align}\label{eq:3sum}
    \lambda_1+\lambda_2&=q-p+o(p+q),  \\ 
    \lambda_1\lambda_2 &\sim -pq .  \label{eq:3prod}
\end{align}
From this, it is easy to see that 
\begin{align}
    \lambda_1&=-p+o(p),  \\
    \lambda_2 &=q+o(q).
\end{align}
(For example, let us see the case $q>p$: Without loss of generality we assume that $\abs{\lambda_1} \leq \abs{\lambda_2}$. From \eqref{eq:3sum} it follows that $\lambda_2 \approx q$, then it follows from \eqref{eq:3prod} that $\lambda_1 \approx -p$. So $\lambda_1=o(\lambda_2)$ therefore from \eqref{eq:3sum} it follows that $\lambda_2 \sim q$, then from \eqref{eq:3prod}   $\lambda_1 \sim -p$. The $q=p$ and $q<p$ cases can be seen in the same way.)

Now the only thing left is to improve the order of the above approximation of the negative eigenvalue $\lambda_1$ to $o(pt)$. For this we will use the ``factorising trick'' again, that is, $-p$ is ``almost'' a root of the cubic equation \eqref{eq:char3}, thus it  ``almost'' can be factorising by $(\lambda_1+p)$ aside from the term on the RHS (remainder) which is ``very'' fast decreasing, 
\[
(\lambda_1+p)\Big(\lambda_1^2-(1+q)\lambda_1 +q-qs+pt\Big)=pt(p+q).
\]
We can write $\lambda_1=-p+\Delta$, where $\Delta=o(p)$, with this,
\[
\Delta\Big((\Delta-p)^2-(1+q)\Delta+(q+p) +pq-qs+pt\Big)=pt(p+q).
\]
Since $\Delta=o(p)$, the second factor on the LHS is $\sim (q+p)$, it follows that $\Delta\sim pt$, so we get that
\begin{equation}
\lambda_1=-p+pt+o(pt).
\end{equation}
Now we use \eqref{eq:lambdaminus} to conclude that
\[
\norm{A(n)}_1=1+p+q-2pt+o(pt) ,
\]
which finishes the proof in the case when the parameters $t,s\in(0,1)$.

Now let's deal with the extreme cases. When $s=0$, the above argument works. When $t=0$, the quadratic equation can factorise,   
\[
(p + \lambda) (\lambda^2-\lambda(1+q) +q(1- s))=0,
\]
that is, the negative eigenvalue is exactly $\lambda_1=-p$, therefore $\norm{A(n)}_1=1+p+q$ which coincides with \eqref{eq:3norm} (if we substitute $t=0$ to the formula).
\end{proof}

%________________________________________________________________PROP 4
\begin{proposition}\label{prop:B} 
Let $B(n)$ be the matrix given in Proposition \ref{prop:LHS} and $r\neq0$. \\If $pt<q$, then
\begin{equation}\label{eq:4norm}
\norm{B(n)}_1=1+p^n+q^n+\abs{p^n-q^n} +o((\min\{p,q\})^n).
\end{equation}
If $pt\geq q$, then
\begin{equation}\label{eq:4norm_pt}
\norm{B(n)}_1=1+p^n+q^n+\abs{p^n-q^n}-2(pt)^n+o((pt)^n).
\end{equation}
\end{proposition}
%______________________________________________________proof
\begin{proof}
The proof is very similar to the proof of Proposition \ref{prop:A}, just a little more complicated. We will use the same notation as there, introduced in Remark \ref{re:notation}. From $r\neq0$ it follows that $t\neq0$ and $s\neq0$ so we have $t,s,r\in(0,1)$.
The 1-norm of $B(n)$ is the sum of the absolute values of the eigenvalues.
Since we know the sum of the eigenvalues ($\Tr B(n)$), just the negative (or just the positive) eigenvalues determine the 1-norm, so we will concern to one of them. (For the easier readability, we do not denote the $n$ dependence of the eigenvalues, i.e., $\lambda_i:=\lambda_i(n)$.) 
The eigenvalues of $B(n)$ are the roots of its characteristic polynomial
\begin{align}  
f_n(\lambda):&=\det(B(n)-\lambda I),
\end{align}
which gives the following quartic equation
\begin{align}\label{eq:quartic}
\lambda^4+a_3\lambda^3+a_2\lambda^2+a_1\lambda+a_0 =0 , && 
\end{align}
where
\begin{align}  
&a_3=-1+2p-2q ,\\   
&a_2= p (-2 - 3 q + t)+q (2 + q - s) +p^2 ,\\   
&a_1=p^2 (-1 - q + t)+q^2 (-1 + s) + p q (3 + q + 2 r - 2 s - 2 t)  ,\\   
&a_0=pq(p-q)(1-t-s+r) .
\end{align}
Of course, we could solve the cubic equation exactly, but this would be a much more difficult way. Instead of that, we take advantages of that we only need the eigenvalues in a certain order of precision. Observe that $f_n(\lambda) \rightarrow \lambda^3 (\lambda-1)$, therefore
\begin{align*} 
&\lambda_1 \to 0 ,  \\ 
&\lambda_2 \to 0 ,  \\ 
&\lambda_3 \to 0 ,  \\ 
&\lambda_4 \to 1 ,  
\end{align*}
so $\lambda_4=1+o(1)$. In the first step, we will improve the order of the approximation of this eigenvalue, what will help us to determine the other ones. For this we will use the ``factorising trick'' as before, that is, $1$ is ``almost'' a root of the quartic equation \eqref{eq:quartic} thus it can ``almost'' be factorising by $(\lambda_4-1)$ aside from the (``fast decreasing'') extra term
\begin{align*}
    (\lambda_4-1)&\Big(\lambda_4^3+2(p-q) \lambda_4^2+(p^2 - 3 p q + q^2 - q s + p t) \lambda_4 + \\
&+ p q^2-p^2 q  + 2 p q r - q s - 2 p q s + q^2 s + p t + p^2 t -  2 p q t \Big) 
+ p t(1-q) -q s(1+p) =0 ,
\end{align*}
from which 
\[
(\lambda_4-1)\Big( \lambda_4^3 +(2 p - 2 q)\lambda_4^2 +(p^2 - 3 p q + q^2 - q s + p t)\lambda_4 +o(p+q)  \Big)\sim qs-pt .
\]
We can write $\lambda_4=1+\Delta$ where $\Delta=o(1)$, with this, we have
\[
\Delta\Big((1+\Delta)^3+((1+\Delta)^2+(1+\Delta))o(1)\Big)\sim qs-pt .
\]
Since $\Delta=o(1)$, the second factor on the LHS is $\sim 1$, it follows that $\Delta\sim qs-pt$ so we get that
\begin{equation}\label{eq:one}
\lambda_4=1+qs-pt+o(qs-pt) ,
\end{equation}
what particularly means that $\lambda_4=1+o(p+q)$. 

Now we turn our focus to the other three eigenvalues, and we will prove the proposition considering the three cases $p=q$, $p<q$ and $p>q$ separately. We will use the following (Vieta's) formulas several times 
 \begin{align} \label{eq:sum}
 \lambda_1+\lambda_2+\lambda_3+\lambda_4&=-a_3 =1-2p+2q ,\\ \label{eq:prod}
 \lambda_1\lambda_2\lambda_3\lambda_4&=a_0 =pq(p-q)(1-t-s+r) ,\\  \label{eq:vieta}
 \lambda_1\lambda_2\lambda_3+\lambda_1\lambda_2\lambda_4+\lambda_1\lambda_3\lambda_4+\lambda_2\lambda_3\lambda_4&=-a_1\sim p^2+q^2-3pq .
\end{align}
Obviously, $\Tr B(n)=-a_3$ and $\det B(n)=a_0$.

%_______________________________________________________________________________p=q
\emph{Case} $p=q$: (so obviously $pt<q$) In this case, $a_0=0$, so we have a zero eigenvalue
\begin{equation}\label{eq:zero}
    \lambda_1=0 ,
\end{equation}
and \eqref{eq:one} means that $\lambda_4=1+o(p)$. With these results, \eqref{eq:sum} and \eqref{eq:vieta} takes the following form
\begin{align*} 
\lambda_2+\lambda_3 &=o(p)  ,\\ 
\lambda_2 \lambda_3 &\sim -p^2  .
\end{align*}
From this, it is easy to see that
\begin{align}
    \lambda_2&=-p+o(p) , \\
    \lambda_3 &=p+o(p) .
\end{align}
So $\lambda_2$ is the only negative eigenvalue. With this, the 1-norm is
\[
\norm{B(n)}_1=\Tr (B(n))-2\lambda_2=1+2p+o(p) ,
\]
which is the same as \eqref{eq:4norm} so we finished the proof in this case.

%_______________________________________________________________________________p<q
\emph{Case} $p<q$: (so obviously $pt<q$)
In this case, \eqref{eq:one} means that $\lambda_4=1+o(q)$. 
With this, \eqref{eq:sum} and \eqref{eq:prod} takes the following form
\begin{align} \label{eq:4sum}
\lambda_1+\lambda_2+\lambda_3 &\sim 2q  ,\\ 
\lambda_1\lambda_2 \lambda_3 &\sim -pq^2  .\label{eq:4prod}
\end{align}
Without loss of generality, we assume that $\abs{\lambda_1} \leq \abs{\lambda_2} \leq \abs{\lambda_3}$. From \eqref{eq:4sum} it follows that $\lambda_3 \approx q$. With this, from \eqref{eq:4prod}, it follows that $\lambda_1=o(q)$. 
With these results for $\lambda_1$ and  $\lambda_3$\,, \eqref{eq:vieta} takes the following form (since: 1. term $<$ 3. term,  2. term $<$ 4. term)
\begin{equation}\label{eq:lambda123}
(\lambda_1+\lambda_2)\lambda_3\sim q^2 .
\end{equation}
Since $\lambda_3\approx q$, it follows that $\lambda_2\approx q$ and since $\lambda_1=o(q)$, \eqref{eq:4sum} and \eqref{eq:lambda123} take the following form
\begin{align}  
\lambda_2+\lambda_3  &\sim 2q , \\ 
\lambda_2 \lambda_3 &\sim q^2 .
\end{align}
From this, $\lambda_2\sim \lambda_3\sim q$, and by \eqref{eq:4prod}, $\lambda_1\sim-p$. So we have
\begin{align}\label{eq:pq}
\lambda_1   &\sim -p  , \nonumber\\ 
\lambda_2   &\sim q  , \\ 
\lambda_3 &\sim q  . \nonumber
\end{align}
That is, $\lambda_1=-p+o(p)$ is the only negative eigenvalue. With this, the 1-norm is
\[
\norm{B(n)}_1=\Tr (B(n))-2\lambda_1=1+2q+o(p)
\]
which is the same as \eqref{eq:4norm} so we finished the proof in this case.

%_______________________________________________________________________________p>q
\emph{Case} $p>q$: In this case both $pt<q$ and $pt \geq q$ are possible.
It can be shown that
\begin{align}\label{eq:p>q}
\lambda_1   &\sim -p ,\nonumber\\ 
\lambda_2   &\sim -p  ,\\ 
\lambda_3 &\sim q  ,\nonumber
\end{align}
in the very same way as we showed \eqref{eq:pq} (so we omit the proof of it). So, this time we have two positive and two negative eigenvalues. We take advantages of that we have already known the two positive eigenvalues  $\lambda_1$ and $\lambda_2$ with the required precision (from  \eqref{eq:p>q} and \eqref{eq:one}).

First, we consider the case $q> pt$. By \eqref{eq:p>q} and \eqref{eq:one} we have that
\begin{align*}
\lambda_3 &=q+o(q) , \\
\lambda_4 &=1+o(q) .
\end{align*}
With this, the 1-norm is
\[
\norm{B(n)}_1=-\Tr (B(n))+2(\lambda_3+\lambda_4)=1+2p+o(q),
\]
which is the same as \eqref{eq:4norm}, so we finished the proof in this case.

Now we consider the case when $q\leq pt$. By \eqref{eq:p>q} and \eqref{eq:one}, we have that
\begin{align*}
\lambda_3 &=q+o(q)=q+o(pt) , \\
\lambda_4 &=1-pt+o(pt) .
\end{align*}
With this, the 1-norm is
\[
\norm{B(n)}_1=-\Tr (B(n))+2(\lambda_3+\lambda_4)=1+2p-2pt+o(pt) ,
\]
which is the same as \eqref{eq:4norm_pt}, so we finished the proof in this case.
So we finished the proof of Proposition \ref{prop:B}.
\end{proof}
Now we have everything to give the proof of Theorem \ref{th:mywork}.

%_______________________________________________________________________________ proof of Th
\newpage
\begin{proof}[Proof of the Theorem \ref{th:mywork}]

By Remark \ref{re:th} and Proposition \ref{prop:RHS}, it is enough to show that
\begin{equation}\label{eq:need}
    \overline{p}_e(\vec \varrho,\, \vec\sigma_1+\vec\sigma_2) =\log\Big(\max\Big\{R\min\{p,q\},pt\Big\}\Big) .
\end{equation}
We will consider the case $r=0$ and $r \neq 0$ separately.

\emph{Case} $r=0$: (in this case, $R$ could be $0,1,...$) By Proposition \ref{prop:LHS} and \ref{prop:A}, it follows immediately that
\begin{align*}
    \overline{p}_e(\vec \varrho,\, \vec\sigma_1+\vec\sigma_2) 
    &= \limsup_{n \to \infty} \frac{1}{n} \log \Bigg(2\parr[\Big]{R\min\{p,q\}}^n+2(pt)^n +o((pt)^n) \Bigg) .
\end{align*} 
Since the ordinary limit exists, it is equal to the limit superior. The factor 2 disappears by the limit (cf. Remark \ref{re:without}), and we get the RHS of \eqref{eq:need}.

\emph{Case} $r\neq0$: (then obviously $R\neq 0$). By Proposition \ref{prop:LHS} and \ref{prop:B}, it follows immediately that, if $pt<q$,
\begin{align*}
    \overline{p}_e(\vec \varrho,\, \vec\sigma_1+\vec\sigma_2) 
    &= \limsup_{n \to \infty} \frac{1}{n} \log \Bigg(2\parr[\Big]{R\min\{p,q\}}^n+o((\min\{p,q\})^n) \Bigg)
\end{align*} 
which coincides with the RHS of \eqref{eq:need}. If $pt\geq q$, it follows that
\begin{align*}
    \overline{p}_e(\vec \varrho,\, \vec\sigma_1+\vec\sigma_2) 
    &= \limsup_{n \to \infty} \frac{1}{n} \log \Bigg(2\parr[\Big]{R\min\{p,q\}}^n+2(pt)^n +o((pt)^n) \Bigg) ,
\end{align*} 
which gives the RHS of \eqref{eq:need}, as we have seen above.
We finished the proof of Theorem \ref{th:mywork}.
\end{proof}

\begin{remark}
The reader may wonder what causes the slight difference between Proposition \ref{prop:A} and \ref{prop:B}, that is, why we have to investigate separate cases in Proposition \ref{prop:B}, while in Proposition \ref{prop:A} we have not. It is because of the following. Observe that in the case $r=0$, the parameter $R$ could be $0$ (independently of the others), so the RHS of \eqref{eq:need} could be $\log (pt)$ (independently of $p,q,t,r,s$). Therefore we always have to get the term $-2(pt)^n$ in the 1-norm of $A(n)$ (in the view of the above proof). In contrast, in the case $r\neq0$, the parameter $R$ cannot be $0$, so we have to investigate separate cases in the case of $B(n)$.
\end{remark}

\pagebreak

%\subfile{sections/5_appendix}
\printbibliography
\addcontentsline{toc}{section}{\protect\numberline{}References}

\begin{textblock*}{5.3cm}(10cm,22cm) % {block width} (coords) 
   And finally...\\[3pt] \textit{"The Answer to the Ultimate Question of Life, The Universe,\\ and Everything."}
\end{textblock*}

\end{document}